\documentclass[12pt,a4paper]{llncs} 

\usepackage{fullpage}
\usepackage{amssymb,amsmath,stmaryrd,mathrsfs,nicefrac}
\usepackage[round]{natbib}

\usepackage{lscape}

\usepackage[mathscr]{eucal}
\usepackage{pstricks,pst-node,pst-text,pst-3d}
\usepackage{graphicx}

\newtheorem{Theorem}{Theorem}
\newtheorem{Lemma}{Lemma}

\def\cB{\mathcal{B}}
\def\cC{\mathcal{C}}

\def\cS{\mathcal{S}}
\def\cT{\mathcal{T}}
\def\ub{\underline{b}}
\def\ob{\overline{b}}
\def\uB{\underline{B}}
\def\oB{\overline{B}}
\def\uBp{\underline{B'}}
\def\oBp{\overline{B'}}
\def\supp{\mathrm{supp}}

\def\BB{{\mathbb B}}
\def\RR{{\mathbb R}}
\def\OWA{\mathrm{OWA}}

\def \srightarrow{\stackrel{1}{\rightarrow}}

\begin{document}

\title{A model of anonymous influence \\ with anti-conformist agents}

\author{Michel Grabisch\inst{1}\thanks{Corresponding author.} \and Alexis Poindron\inst{2} \and Agnieszka Rusinowska\inst{3}}

\institute{Paris School of Economics, Universit\'{e} Paris I Panth\'{e}on-Sorbonne  \\
Centre d'Economie de la Sorbonne, 106-112 Bd de l'H\^{o}pital, 75647 Paris Cedex 13, France \\
\email{michel.grabisch@univ-paris1.fr} 
\and 
Universit\'{e} Paris I Panth\'{e}on-Sorbonne, 
Centre d'Economie de la Sorbonne \\
\email{alexis.poindron@laposte.net}
\and
CNRS, Paris School of Economics, Centre d'Economie de la Sorbonne
\\
\email{agnieszka.rusinowska@univ-paris1.fr}}

\date{Version of \today}

\maketitle

\begin{abstract}
We study a stochastic model of anonymous influence with conformist and
anti-conformist individuals. Each agent with a `yes' or `no' initial opinion on
a certain issue can change his opinion due to social influence. We consider
anonymous influence, which depends on the number of agents having a certain
opinion, but not on their identity. An individual is conformist/anti-conformist
if his probability of saying `yes' increases/decreases with the number of `yes'-agents. 
We focus on three classes of aggregation rules (pure conformism, pure
anti-conformism, and mixed aggregation rules) and examine two types of society
(without, and with mixed agents).  For both types we provide a complete
qualitative analysis of convergence, i.e., identify all absorbing classes and
conditions for their occurrence. Also the pure case with infinitely many
individuals is studied. We show that, as expected, the presence of
anti-conformists in a society brings polarization and instability: polarization
in two groups, fuzzy polarization (i.e., with blurred frontiers), cycles,
periodic classes, as well as more or less chaotic situations where at any time
step the set of `yes'-agents can be any subset of the society. Surprisingly, the
presence of anti-conformists may also lead to opinion reversal: a majority group
of conformists with a stable opinion can evolve by a cascade phenomenon towards
the opposite opinion, and remains in this state.

\end{abstract}


\vspace{5 mm}


\section{Introduction}

This paper is devoted to anti-conformism in the framework of opinion formation with anonymous influence. Despite the fact that anti-conformism plays a crucial role in many social and economic situations, and can naturally explain human behavior and various dynamic processes, this phenomenon did not receive enough attention in the literature. 

The seminal work of \cite{deg74} and some of its extensions consider a non-anonymous positive influence in which agents update their opinions by using a weighted average of opinions of their neighbors. However, in many situations, like opinions and comments given on the internet,
the identity of the agents is not known, or at least, there is no clue on the
reliability or kind of personality of the agents. Therefore, agents can be
considered as anonymous, and influence is merely due to the \textit{number} of
agents having a certain opinion, not their identity. 
 \cite{foe-gra-rus13} investigate such an anonymous social influence, but restrict their attention to the conformist behavior. They depart from a general framework of influence based on aggregation functions (\cite{gra-rus13}), where every individual updates his opinion by aggregating the agents' opinions which determines the probability that his opinion will be `yes' in the next period. Instead of allowing for arbitrary aggregation functions, \cite{foe-gra-rus13} consider anonymous aggregating. However, both frameworks of \cite{gra-rus13} and \cite{foe-gra-rus13} cover only positive influence (imitation), since by definition aggregation functions are nondecreasing, and hence cannot model anti-conformism. 

\bigskip

Our aim is to study opinion formation under anonymous influence in societies
with conformist and anti-conformist individuals. We focus on three classes of
aggregation rules that can be used by the agents when revising their opinions:
pure conformism, pure anti-conformism, and mixed aggregation rules. As a
consequence, we distinguish three types of agents: pure conformists who are
more likely to say `yes' when there are more agents who said `yes' in the last
period, pure anti-conformists who are less likely to do so, and mixed agents.
All purely conformist agents are assumed to share the same minimum number of
`yes' for them to say `yes' with positive chance and the same minimum number of
`yes' for them to say `yes' with chance of 1. Similarly, all purely
anti-conformist agents share the same minimum number of `yes' for them to say
`no' with positive chance and the same minimum number of `yes' for them to say
`no' with chance of 1. We call these four parameters the influenceability parameters.
We consider two types of a society: without mixed
agents, and containing mixed agents. For both cases we provide a complete
qualitative analysis of convergence, i.e., we identify all absorbing classes and
conditions for their occurrence. This full characterization of the absorbing
classes is based on the size comparison among the four influenceability
parameters and the number of conformists and anti-conformists in the
society.  

\bigskip

Our findings bring precise answers to the following fundamental questions: {\it
  What is the impact of the presence of anti-conformists on a society being
  mainly conformist? Is a chaotic or unstable situation possible? Is opinion
  reversal possible?} 

The exact description of the impact is done through our
main Theorems~\ref{th:main} and \ref{th:main2}, giving all possible absorbing
classes, that is, all possible states of the society in the long run, and conditions of their existence. The complexity of these results, however, asks
for a further analysis which would extract the main trends. We have conducted
such an analysis, supposing that the size of the society is very large, and
considering several typical situations, e.g., the same influenceability parameters
among the conformists and anti-conformists, a small proportion of
anti-conformists, etc. This permits to answer the two other questions. 

About the
possibility of a chaotic or unstable situation, without much surprise, the
answer is positive. We have distinguished however several types of unstable
situations, ranked in increasing order of unstability: fuzzy polarization,
cycle, fuzzy cycle and chaos. Polarization is the well-known phenomenon where a
part of the society converges to `yes' and the other part to `no'. Obviously,
this situation can arise here, with conformists and anti-conformists having
opposite opinion, provided the proportion of the latter is not too high. Fuzzy
polarization means that instead of having two stable groups, there is a kind of
oscillation around the groups of conformists and anti-conformists. When the set
of `yes'-agents evolve according to a periodic sequence of subsets of the
society, we speak of a cycle. For instance, it is easy to see that the cycle
$N^a,N^c,N^a,\ldots$, where $N^a, N^c$ are respectively the set of
anti-conformists and the set of conformists, can arise, provided the number of
anti-conformists is large enough. A fuzzy cycle is nothing other than a
periodic class in the theory of Markov chains. That is, there is a cycle of
{\it collections} of subsets of the society, and at each time step, a subset is picked
at random in the collection under consideration. Finally, chaos means that at
each time step, any subset of the society can be the set of `yes'-agents. 

Finally, is opinion reversal possible? First, some explanation is necessary. In
a purely conformist society, there is quick convergence to a consensus, either
on `yes' or on `no'. Therefore, the opinion remains constant for ever, and no
change can occur. Opinion reversal would mean that, starting from a state where
a large majority of the society has a stable opinion, there would be an
evolution leading for that majority group to the opposite opinion (cascade
phenomenon). This is exactly what is most feared by, e.g., politicians during
some election, or any leader governing some society of
individuals. Surprisingly, such a phenomenon can occur, even with a small number
of anti-conformists, under certain circumstances that we describe
precisely. Hence, contrarily to the intuition which tells us that introducing
anti-conformists is simply introducing unstability and chaos, we have shown that
it is quite possible to manipulate, by a suitable choice of the influenceability
parameters and the proportion of anti-conformists, the final opinion of the
conformists. We consider this result as one of the main findings of the paper.

\bigskip

Our framework is suitable for many natural applications. It can explain various phenomena like stable and persistent shocks, large fluctuations, stylized facts in the industry of fashion, in particular its intrinsic dynamics, booms and burst in the frequency of surnames, etc. If fashion were only a matter of conformist imitation in an anonymous framework, there would be no trends over time. 
Anti-conformism and anti-coordination when individuals have an incentive to differ from what others do can also be detected, e.g., in organizational settings. For example, the choice of a firm to go compatible or not with other firms can be seen as a problem of anti-conformism. Anti-coordination can be optimal when adopting different roles or having complementary skills is necessary for a successful interaction or realization of a task in a team.  

\bigskip

The rest of the paper is structured as follows. In Section \ref{sec:model} we introduce the model of anonymous influence with anti-conformist agents and distinguish between two kinds of a society: pure case (containing only pure anti-conformists and pure conformists) and mixed case (including also mixed agents). We also explain how the present paper extends and differs from our previous work on conformism. The convergence analysis of both cases as well as of the pure case with the number of agents tending to infinity is provided in Section \ref{sec:convergence}.  In Section \ref{sec:related-literature} we deliver a brief overview of the related literature. We mention some concluding remarks in Section \ref{sec:conclusion}. 
The proof of our main results on the possible absorbing classes in the model is given in the appendix.

\section{The model}\label{sec:model} 

\subsection{Description of the model}\label{subsec:description}

We consider a society $N$ with $|N|=n$ agents who make a `yes' or `no' decision
on some issue (binary choice).  Initially every agent has an opinion on that
issue, but by knowing the opinion of the others or by some social interaction,
in each period the agent may modify his opinion due to mutual influence. In
other words, there is an evolution in time of the opinion of the agents, which
may or may not stop at some stable state of the society.

We define the {\it state} of the society at a given time as the set $S\subseteq
N$ of agents whose opinion is `yes'. As usual, the cardinality of a set is denoted
by the corresponding lower case, e.g., $s=|S|$.  Our fundamental assumption is
that the evolution of the state is ruled by a homogeneous Markov chain, that is,
the state evolves at discrete time steps, the new opinions depend only on the
opinions among the society in the last period, and the transition matrix giving
the probability of all possible transitions from one state to another is
constant over time.

We study the opinion formation process of the society, where the updating of opinion
relies only on how many agents said `yes' and how many said `no' in the previous
period, disregarding who said `yes' and who said `no'. For this reason we call
this opinion formation process {\it anonymous}. Both conformist and anti-conformist behaviors
are allowed, i.e., agents can revise their opinions by following the trend as
well as in a way contrary to the trend.

We now formalize the previous ideas. An {\it (anonymous) aggregation rule}
describes for a given agent how the opinions of the other agents are aggregated
in order to determine the updated opinion of this agent. Specifically, it is a
mapping $p$ from $\{0,1,\ldots,n \}$ to $[0,1]$, assigning to any $0\leq s\leq
n$, representing the number of agents saying `yes', a quantity $p(s)$ which is
the probability of saying `yes' at next time step (and consequently, $1-p(s)$ is
the probability of saying `no').  We focus on three classes of aggregation
rules:
\begin{equation}\label{eq:1new}
A^c = \{ p \mid p \text{ is nondecreasing and satisfies } p(0)=0 \text{ and } p(n)=1 \} \quad \text{{\it (pure conformism)}}
\end{equation}
\begin{equation}\label{eq:2new}
A^a = \{ p \mid p  \text{ is nonincreasing and satisfies } p(0)=1 \text{ and }
p(n)=0\} \quad \text{{\it (pure anti-conformism)}}
\end{equation}
and {\it mixed aggregation rules}:
\begin{equation}\label{eq:3new}
A^m = \{ p \mid p = \alpha q^c + (1-\alpha) q^a \text{ with } \alpha \in \left]0,1\right[, q^c \in A^c, q^a \in A^a \}.
\end{equation}
Note that $p\in A^m$ is not necessarily monotone, and that $p(0)>0$ and
$p(n)<1$. 

Let $p_i$ be agent $i$'s (anonymous) aggregation rule. Supposing that the update of opinion is done
independently accross the agents, the probability of transition from a state $S$ to a state $T$ is
\begin{equation}\label{eq:1}
\lambda_{S,T} = \prod_{i\in T}p_i(s)\prod_{i\not\in T}(1-p_i(s)).
\end{equation}
Observe that, as the aggregation rule is anonymous, the probability of
transition to $T$ is the same for all states $S$ having the same size $s$. 

We distinguish between three types of agents. We say that agent $i$ is \textit{purely conformist} if $p_i \in A^c$, \textit{purely anti-conformist} if $p_i\in A^a$, and is a \textit{mixed agent} if $p_i \in A^m$. 
The society of agents is partitioned into
\[
N = N^c \cup N^a \cup N^m
\]
where $N^c$ is the set of purely conformist agents, $N^a$
the set of purely anti-conformist agents, and $N^m$ is the set
of mixed agents. 
When $N^m=\emptyset$, we call it the \textit{pure case}.

We make the following simplifying assumption: we suppose that all purely
conformist agents, although having possibly different aggregation rules, have in
common the same minimum number
of `yes' for them to say `yes' with positive chance, and the same minimum number
of `yes' for them to say `yes' with probability 1.
Formally, we require that for each $i \neq j$ in $N^c$,
\begin{equation}\label{eq:4new}
\min \{s \mid p_i (s) > 0\} = \min \{s \mid p_j (s) > 0\} =: l^c + 1
\end{equation}
\begin{equation}\label{eq:5new}
\min \{s \mid p_i (s) = 1\} = \min \{s \mid p_j (s) = 1\} =: n - r^c .
\end{equation}
$l^c$ can be interpreted as the (maximum) number of `yes' for which no effect on
the probability of saying `yes' arises, while $r^c$ is the (maximum) number of `no' for which no
effect on this probability is visible (see Figure~\ref{fig:lr}). This assumption
permits to conduct a precise analysis of convergence while remaining reasonable
(it can be considered as a mean field approximation). 
\begin{figure}[htb]
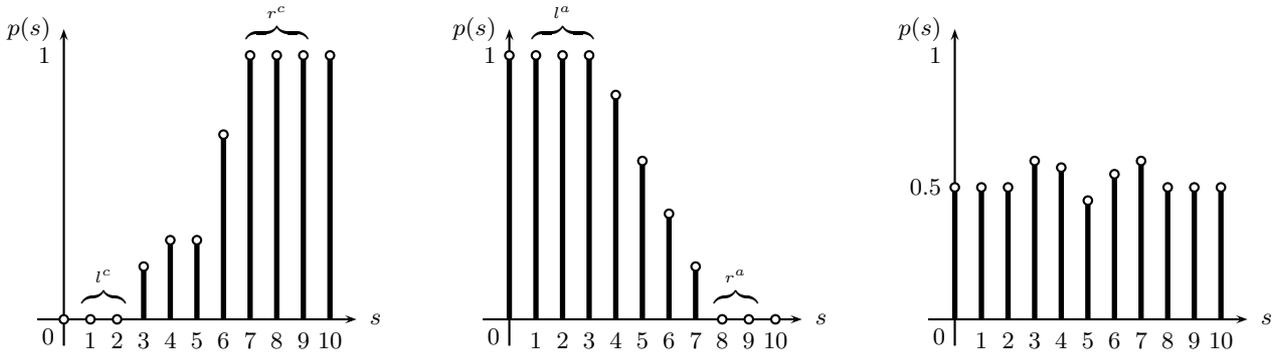

\begin{center}
\psset{unit=0.7cm}
\pspicture(-0.5,-1)(5.5,5)
\psline{->}(-0.5,0)(5.5,0)
\psline{->}(0,-0.5)(0,5.5)
\psline[linewidth=2pt](0.5,0)(0.5,0)
\psline[linewidth=2pt](1,0)(1,0)
\psline[linewidth=2pt](1.5,0)(1.5,1)
\psline[linewidth=2pt](2,0)(2,1.5)
\psline[linewidth=2pt](2.5,0)(2.5,1.5)
\psline[linewidth=2pt](3,0)(3,3.5)
\psline[linewidth=2pt](3.5,0)(3.5,5)
\psline[linewidth=2pt](4,0)(4,5)
\psline[linewidth=2pt](4.5,0)(4.5,5)
\psline[linewidth=2pt](5,0)(5,5)
\pscircle[fillstyle=solid](0,0){0.1}
\pscircle[fillstyle=solid](0.5,0){0.1}
\pscircle[fillstyle=solid](1,0){0.1}
\pscircle[fillstyle=solid](1.5,1){0.1}
\pscircle[fillstyle=solid](2,1.5){0.1}
\pscircle[fillstyle=solid](2.5,1.5){0.1}
\pscircle[fillstyle=solid](3,3.5){0.1}
\pscircle[fillstyle=solid](3.5,5){0.1}
\pscircle[fillstyle=solid](4,5){0.1}
\pscircle[fillstyle=solid](4.5,5){0.1}
\pscircle[fillstyle=solid](5,5){0.1}
\uput[180](0,5.5){$p(s)$}
\uput[180](0,5){1}
\uput[225](0,0){0}
\uput[-90](0.5,0){1}
\uput[-90](1,0){2}
\uput[-90](1.5,0){3}
\uput[-90](2,0){4}
\uput[-90](2.5,0){5}
\uput[-90](3,0){6}
\uput[-90](3.5,0){7}
\uput[-90](4,0){8}
\uput[-90](4.5,0){9}
\uput[-90](5,0){10}
\uput[0](5.5,0){$s$}
\uput[90](0.75,0){$\overbrace{\rule{0.5cm}{0cm}}^{l^c}$}
\uput[90](4,5){$\overbrace{\rule{0.8cm}{0cm}}^{r^c}$}
\endpspicture
\hfill
\pspicture(-0.5,-1)(5.5,5)
\psline{->}(-0.5,0)(5.5,0)
\psline{->}(0,-0.5)(0,5.5)
\psline[linewidth=2pt](0,0)(0,5)
\psline[linewidth=2pt](0.5,0)(0.5,5)
\psline[linewidth=2pt](1,0)(1,5)
\psline[linewidth=2pt](1.5,0)(1.5,5)
\psline[linewidth=2pt](2,0)(2,4.25)
\psline[linewidth=2pt](2.5,0)(2.5,3)
\psline[linewidth=2pt](3,0)(3,2)
\psline[linewidth=2pt](3.5,0)(3.5,1)
\psline[linewidth=2pt](4,0)(4,0)
\psline[linewidth=2pt](4.5,0)(4.5,0)
\psline[linewidth=2pt](5,0)(5,0)
\pscircle[fillstyle=solid](0,5){0.1}
\pscircle[fillstyle=solid](0.5,5){0.1}
\pscircle[fillstyle=solid](1,5){0.1}
\pscircle[fillstyle=solid](1.5,5){0.1}
\pscircle[fillstyle=solid](2,4.25){0.1}
\pscircle[fillstyle=solid](2.5,3){0.1}
\pscircle[fillstyle=solid](3,2){0.1}
\pscircle[fillstyle=solid](3.5,1){0.1}
\pscircle[fillstyle=solid](4,0){0.1}
\pscircle[fillstyle=solid](4.5,0){0.1}
\pscircle[fillstyle=solid](5,0){0.1}
\uput[180](0,5.5){$p(s)$}
\uput[180](0,5){1}
\uput[225](0,0){0}
\uput[-90](0.5,0){1}
\uput[-90](1,0){2}
\uput[-90](1.5,0){3}
\uput[-90](2,0){4}
\uput[-90](2.5,0){5}
\uput[-90](3,0){6}
\uput[-90](3.5,0){7}
\uput[-90](4,0){8}
\uput[-90](4.5,0){9}
\uput[-90](5,0){10}
\uput[0](5.5,0){$s$}
\uput[90](1,5){$\overbrace{\rule{0.8cm}{0cm}}^{l^a}$}
\uput[90](4.25,0){$\overbrace{\rule{0.1cm}{0cm}}^{r^a}$}
\endpspicture
\hfill
\pspicture(-0.5,-1)(5.5,5)
\psline{->}(-0.5,0)(5.5,0)
\psline{->}(0,-0.5)(0,5.5)
\psline[linewidth=2pt](0,0)(0,2.5)
\psline[linewidth=2pt](0.5,0)(0.5,2.5)
\psline[linewidth=2pt](1,0)(1,2.5)
\psline[linewidth=2pt](1.5,0)(1.5,3)
\psline[linewidth=2pt](2,0)(2,2.875)
\psline[linewidth=2pt](2.5,0)(2.5,2.25)
\psline[linewidth=2pt](3,0)(3,2.75)
\psline[linewidth=2pt](3.5,0)(3.5,3)
\psline[linewidth=2pt](4,0)(4,2.5)
\psline[linewidth=2pt](4.5,0)(4.5,2.5)
\psline[linewidth=2pt](5,0)(5,2.5)
\pscircle[fillstyle=solid](0,2.5){0.1}
\pscircle[fillstyle=solid](0.5,2.5){0.1}
\pscircle[fillstyle=solid](1,2.5){0.1}
\pscircle[fillstyle=solid](1.5,3){0.1}
\pscircle[fillstyle=solid](2,2.875){0.1}
\pscircle[fillstyle=solid](2.5,2.25){0.1}
\pscircle[fillstyle=solid](3,2.75){0.1}
\pscircle[fillstyle=solid](3.5,3){0.1}
\pscircle[fillstyle=solid](4,2.5){0.1}
\pscircle[fillstyle=solid](4.5,2.5){0.1}
\pscircle[fillstyle=solid](5,2.5){0.1}
\uput[180](0,5.5){$p(s)$}
\uput[180](0,5){1}
\uput[180](0,2.5){0.5}
\uput[225](0,0){0}
\uput[-90](0.5,0){1}
\uput[-90](1,0){2}
\uput[-90](1.5,0){3}
\uput[-90](2,0){4}
\uput[-90](2.5,0){5}
\uput[-90](3,0){6}
\uput[-90](3.5,0){7}
\uput[-90](4,0){8}
\uput[-90](4.5,0){9}
\uput[-90](5,0){10}
\uput[0](5.5,0){$s$}
\endpspicture
\end{center}
\caption{Typical aggregation rules for conformist agents (left), anti-conformists
  (center), and mixed agents (right) with $n=10$ agents. The latter is obtained by mixing the two
  first ones with $\alpha=0.5$}
\label{fig:lr}
\end{figure}

Similarly, we assume that all purely anti-conformist agents share the same minimum
number of `yes' for them to say `no' with positive chance and the same minimum number of `yes' for
them to say `no' with probability 1, i.e., 
for each $i \not= j$ in $N^a$,
\begin{equation}\label{eq:6new}
\min \{s \mid p_i (s) < 1\} = \min \{s \mid p_j (s) < 1\} =: l^a + 1
\end{equation}
\begin{equation}\label{eq:7new}
\min \{s \mid p_i (s) = 0\} = \min \{s \mid p_j (s) = 0\} =: n - r^a. 
\end{equation}
As above, $l^a$ (respectively, $r^a$) is the maximum number of `yes'
(respectively, `no') for which no effect on the probability of saying `yes' is
visible (see Figure~\ref{fig:lr}).

These assumptions carry over mixed agents: any mixed agent $i$ has an
aggregation rule $p_i$ which is a convex combination with coefficient $\alpha_i$
of a conformist aggregation rule (with fixed $l^c,r^c$) and an anti-conformist
aggregation rule (with fixed $l^a,r^a$). We allow, however, that two mixed
agents have different convex combination coefficients $\alpha_i$. Hence, a mixed
agent $i$ can be seen as an agent who does not have a fixed behavior, but who is
conformist with probability $\alpha_i$ and anti-conformist with probability
$1-\alpha_i$. On Figure~\ref{fig:lr}, we can see that a mixed agent with
$\alpha_i=1/2$ has a very indecisive behavior. 

Based on the above assumptions, in this paper we fully characterize all possible absorbing classes based on size comparison among $l^c,r^c,l^a,r^a$ and the number of conformist and anti-conformist agents in the society.

\subsection{Basic properties of the transition matrix}
We recall that (see, e.g., \cite{kem-sne76,sen06}) for a
Markov chain with set of states $E$ and transition matrix $\Lambda$ and its
associated digraph $\Gamma$, a \textit{class} is a subset $C$ of states such
that for all states $e,f\in C$, there is a path in $\Gamma$ from $e$ to $f$, and
$C$ is maximal w.r.t. inclusion for this property. A class is \textit{absorbing} if for every $e\in C$ there is no arc in $\Gamma$ from $e$
to a state outside $C$. An absorbing class $C$ is \textit{periodic of period $k$}
if it can be partitioned in blocks $C_1,\ldots, C_k$ such that for
$i=1,\ldots,k$, every outgoing arc of every state $e\in C_i$ goes to some state
in $C_{i+1}$, with the convention $k+1=1$. When each $C_1,\ldots, C_k$
  reduces to a single state, one may speak of \textit{cycle of length $k$}, by
  analogy with graph theory.

In our framework, states are subsets of agents and therefore classes are collections
of sets, which we denote by calligraphic letters, like $\cC,\cB$, etc. By
definition, an absorbing class indicates the final state of opinion of the
society. For instance, if an absorbing class reduces to a single state $S$, it means that in the
long run, the society is dichotomous (unless $S=N$ or $S=\emptyset$, in which
case consensus is reached): there is a set of agents $S$ who say `yes' forever,
while the other ones say `no' forever. Otherwise, there are endless transitions
with some probability from one set $S\in \cC$ to another one $S'\in \cC$.

\bigskip

We now study the properties of the transition matrix $\Lambda$, with
entries $\lambda_{S,T}$, $S,T\in 2^N$, where $\lambda_{S,T}$ is given by
(\ref{eq:1}).
Our aim is to find under which conditions one has a possible transition from $S$
to $T$, i.e., $\lambda_{S,T}>0$. From (\ref{eq:1}), we have:
\[
\lambda_{S,T}>0\Leftrightarrow [p_i(s)>0 \ \ \forall i\in T] \ \ \& \ \ [p_i(s)<1\ \ \forall
i\not\in T].
\]

We start with the pure case, i.e., $N^m=\emptyset$. We first observe that $p_i(0)=1$
if $i\in N^a$ and 0 otherwise, and $p_i(n)=1$ if $i\in N^c$ and 0
otherwise. Therefore, we have in any case the sure transitions
\[
\lambda_{\emptyset,N^a}=1, \quad \lambda_{N,N^c}=1.
\]
Moreover, we get immediately from (\ref{eq:4new}) to (\ref{eq:7new}) for any $s\neq 0, n$,
\begin{align}
(i\in N^c)\qquad  & p_i(s) >0 && \Leftrightarrow && s>l^c\label{eq:1a}\\
 & p_i(s) <1 && \Leftrightarrow  &&
  s< n-r^c\label{eq:3a}\\
(i\in N^a)\qquad  & p_i(s) >0 && \Leftrightarrow && s<n-r^a\label{eq:4a}\\
 & p_i(s) <1 && \Leftrightarrow &&  s> l^a\label{eq:6a}.
\end{align}

By combining these conditions and their negation in various ways, one can see
that we can only have transitions to $\emptyset, N, N^a,N^c$ and any of the
subsets or supersets of $N^a,N^c$. A convenient notation here is the interval
notation: for two sets $S\subseteq T$, $[S,T]$ denotes the collection of sets
$K$ such that $S\subseteq K\subseteq T$. For instance, $[\emptyset, N^a]$ and
$[N^a,N]$ denote respectively the collection of subsets of $N^a$ and the
collection of supersets of $N^a$. Table~\ref{tab:1} summarizes the possible
transitions, adding also those from $S=\emptyset$ and $S=N$.
\begin{table}[htb]
\begin{center}
\begin{tabular}{|c|c|c|c|}\hline
 & $0\leq s\leq l^c$ & $l^c<s<n-r^c$ &    $n-r^c\leq s\leq n$\\ \hline
$0\leq s\leq l^a$ & $N^a$ & $T\in[N^a,N] $& $N$\\ \hline
$l^a<s<n-r^a$ & $T\in[\emptyset,N^a]$ & $T\in 2^N$& $T\in[N^c,N]$\\ \hline
$n-r^a\leq s\leq n$ & $\emptyset$ & $T\in[\emptyset,N^c]$& $N^c$\\ \hline
\end{tabular}
\vspace{2 mm}
\caption{Possible transitions from $S\in 2^N$ in the pure case}
\label{tab:1}
\end{center}
\end{table}

Let us introduce $Z=(l^c,r^c,l^a,r^a)$ and write $p_i^Z$ to emphasize the dependency of $p_i$ on
these parameters (and similarly for $\lambda_{S,T}$). Equations (\ref{eq:1a}) to
(\ref{eq:6a}) show striking symmetries, in particular, when interchanging conformists and
anti-conformists. $Z$ being
given, we introduce the reversal of $Z$, $Z^\partial :=(r^c,l^c,r^a,l^a)$, and the interchange of $Z$,
$Z'=(l^a,r^a,l^c,r^c)$ which amounts to interchanging conformists with
anti-conformists. Considering these operations, we observe the following
symmetries:
\begin{enumerate}
\item Interchange:
\begin{align*}
p_i^Z(s)>0 \text{ for } i\in N^c && \Leftrightarrow && p_i^{Z'}(s)<1 \text{ for }
i\in N^a\\
p_i^Z(s)<1 \text{ for } i\in N^c && \Leftrightarrow && p_i^{Z'}(s)>0 \text{ for }
i\in N^a
\end{align*}
(idem with $N^a$, $N^c$ exchanged)
\item Reversal:
\begin{align*}
p_i^Z(s)>0 \text{ for } i\in N^c && \Leftrightarrow && p_i^{Z^\partial}(n-s)<1 \text{ for }
i\in N^c\\
p_i^Z(s)<1 \text{ for } i\in N^c && \Leftrightarrow && p_i^{Z^\partial}(n-s)>0 \text{ for }
i\in N^c
\end{align*}
(idem with $N^a$, $N^c$ exchanged)
\item Interchange and reversal:
\begin{align*}
p_i^Z(s)>0 \text{ for } i\in N^c && \Leftrightarrow && p_i^{(Z^\partial)'}(n-s)>0 \text{ for }
i\in N^a\\
p_i^Z(s)<1 \text{ for } i\in N^c && \Leftrightarrow && p_i^{(Z^\partial)'}(n-s)<1 \text{ for }
i\in N^a
\end{align*}
(idem with $N^a$, $N^c$ exchanged)
\end{enumerate}
The second case is of particular interest and leads to the following lemma.
\begin{Lemma}[\bf symmetry principle]\label{lem:sym}
Let $S,T\in 2^N$ and $Z=(l^c,r^c,l^a,r^a)$. The following equivalence holds:
\[
\lambda^Z_{S,T}>0\Leftrightarrow \lambda^{Z^\partial}_{N\setminus S,N\setminus T}>0.
\]
\end{Lemma}
\begin{proof}
Letting $\lambda^Z_{S,T}>0$ means that for every $i\in N\setminus T$, $0\leq
p^Z_i(s)<1$, and for every $i\in T$, $0<p^Z_i(s)\leq 1$. Using the equivalences
in (ii), we find that for every $i\in N\setminus T$,
$0<p^{Z^\partial}_i(n-s)\leq 1$ and for every $i\in T$, $0\leq
p^{Z^\partial}_i(n-s)<1$. But this means that
$\lambda^{Z^\partial}_{N\setminus S,N\setminus T}>0$.
\qed\ 
\end{proof}


\bigskip

Next we consider the mixed case, assuming $p_i = \alpha_iq^c + (1-\alpha_i)q^a$
with $q^c\in A^c$, $q^a\in A^a$ and fixed $l^c,r^c,l^a,r^a$. 
We can easily derive the conditions for $p_i(s)$ to be 0 or
1, using (\ref{eq:1a}) to (\ref{eq:6a}). For any $0<s<n$, we obtain
\begin{align}
(i\in N^m) \qquad & p_i(s)=0 && \Leftrightarrow && n-r^a\leq s\leq l^c\label{eq:7}\\
\qquad & p_i(s)=1 && \Leftrightarrow && n-r^c\leq s\leq l^a,\label{eq:8}
\end{align}
and $0<p_i(s)<1$  for all other cases. 
Finally, if $S=\emptyset$, then $\lambda_{S,T}>0$ for every $T\in[N^a,N^a\cup
  N^m]$, and if $S=N$, then $\lambda_{S,T}>0$ for every $T\in [N^c,N^c\cup N^m]$. Table \ref{tab:2} presents possible transitions from $S\in 2^N$ in the mixed case.
\begin{table}[htb]
\begin{center}
\begin{tabular}{|c|c|c|c|}\hline
 & $0\leq s\leq l^c$ & $l^c<s<n-r^c$ &    $n-r^c\leq s \leq n$\\ \hline
$0\leq s\leq l^a$ & $T\in[N^a,N^a\cup N^m]$ & $T\in[N^a,N] $& $N$\\ \hline
$l^a<s<n-r^a$ & $T\in[\emptyset,N^a\cup N^m]$ & $T\in 2^N$& $T\in[N^c,N]$\\ \hline
$ n-r^a\leq s\leq n$ & $\emptyset$ & $T\in[\emptyset,N^c\cup N^m]$& $T\in [N^c,N^c\cup N^m]$\\ \hline
\end{tabular}
\vspace{2 mm}
\caption{Possible transitions from $S\in 2^N$ in the mixed case}
\label{tab:2}
\end{center}
\end{table}


\subsection{Relation with the anonymous model of conformity}\label{sec:anoinf}

We recall the anonymous model of conformism (\cite{foe-gra-rus13}) and show
that it is equivalent to our class of conformist aggregation rules. By doing
this we show that the present model is a natural extension of
\cite{foe-gra-rus13}.

The assumption that agents modify their opinions in a Markovian way is basically that underlying \cite{gra-rus13}. As the number
of states is $2^n$, the size of the transition matrix is $2^n\times 2^n$. In
order to avoid this exponential complexity, like in the present paper the latter reference uses a simple
mechanism to generate the transition matrix, based on aggregation
functions, that is, nondecreasing mappings $A:[0,1]^n\rightarrow[0,1]$
satisfying $A(0,\ldots,0)=0$ and $A(1,\ldots,1)=1$. Specifically, supposing that
agent $i$ aggregates opinions by the function $A_i$, the probability that agent
says `yes' at next time step, given that $S$ is the set of agents saying `yes'
at present, is given by 
\begin{equation}\label{eq:Agg}
p_i(S) = A_i(1_S),
\end{equation}
where $1_S$ is the indicator function of $S$, i.e., $1_S(i)=1$ iff $i\in S$ and
0 otherwise. This model is presented and studied in general in
\cite{gra-rus13}.

The most common example of aggregation function, used, e.g., in
\citet{deg74}, is the weighted arithmetic mean
\begin{equation}
A_i(x)=\sum_{j=1}^nw^i_jx_j,
\end{equation}
where $x=(x_1,\ldots,x_n)$ and the $w^i_j$'s are weights on the entries,
satisfying $w^i_j\geq 0$ and $\sum_{j=1}^nw^i_j=1$. The weight $w^i_j$ represents to
which extent agent $i$ puts confidence on the opinion $x_j$ of agent $j$. It depicts a
situation where every agent knows the identity of every other agent, and is able
to assess to which extent he trusts or agrees with the opinion or personal tastes
of others.

\cite{foe-gra-rus13} investigate the model of conformism with anonymous social influence, which depends only on the \textit{number} of
agents with a certain opinion, not on their identity. They use the \textit{ordered weighted averages} (commonly called OWA operators, \cite{yag88}), which are the unique anonymous aggregation functions:
\begin{equation}\label{eq:2}
\OWA_w(x)=\sum_{j=1}^nw_jx_{(j)},
\end{equation}
where the entries $x_1,\ldots,x_n$ are rearranged in decreasing order:
$x_{(1)}\geq x_{(2)}\geq\cdots\geq x_{(n)}$, and $w=(w_1,\ldots,w_n)$ is the
weight vector defined as above. Hence, the weight $w_j$ is not
assigned to agent $j$ but to rank $j$, and thus permits to model
quantifiers. For example, taking $w_1=1$ and all other weights being 0 models
the quantifier ``there exists''. Indeed, it is enough to have one of the entries
being equal to 1 to get 1 as output. In our context, it means that only one
agent saying `yes' is enough to make your opinion being 'yes' for
sure. Similarly, ``for all'' is modeled by $w_n=1$ and all other weights being
0, and means that you need that all agents (including you) say `yes' to be sure
to continue to say `yes'. Intermediate situations can be modeled as
well. 

\medskip 

For the anonymous model of conformism, there exist two types of absorbing classes (\cite{foe-gra-rus13}):
\begin{enumerate}
\item any single state $S\in 2^N$ (including the consensus states $\emptyset$ and
$N$);
\item union of intervals of the type $[S,S\cup K]$, where $S,K\neq\emptyset,N$
  (recall that $[S,S\cup K]=\{T\in 2^N\mid S\subseteq T\subseteq S\cup K\}$). 
\end{enumerate}  
For the second case, when the absorbing class is reduced to a single interval
$[S,S\cup K]$, it depicts a situation in the long run where agents in $S$ say
`yes', those outside $S\cup K$ say `no', and those in $K$ oscillate between
`yes' and `no' forever. Interestingly, no periodic class can occur, although
in general for arbitrary aggregation functions cycles can occur
(\cite{gra-rus13}).  

\bigskip

We now establish the relation with our framework. Supposing that agent $i$
aggregates opinions anonymously by $\OWA_{w^i}$ with weight vector $w^i$, we have
from (\ref{eq:Agg}) and (\ref{eq:2}), for any $S\subseteq N$:
\[
p_i(S) = \OWA_{w^i}(1_S) = \sum_{j=1}^sw^i_j =:p_i(s).
\] 
We see that $p_i(s)$ is an aggregation rule in the sense of Section~\ref{subsec:description}, and
that, due to the properties of $w^i$, $p_i\in A^c$. Conversely, for any $p\in
A^c$, we can define $w\in \RR^n$ by $w_j=p(j)-p(j-1)$, $j=1,\ldots,n$, and by
properties of $p\in A^c$, $w$ is a well-defined weight vector s.t. $p\equiv
\OWA_w$. The introduction of the classes of aggregation rules $A^a,A^m$ are
therefore natural generalizations of the conformist model of \citet{foe-gra-rus13}.

\bigskip

Note that the numbers $l^c,r^c,l^a,r^a$ defined in (\ref{eq:4new}) - (\ref{eq:7new}) correspond to the numbers of left and right zeroes in the weight vectors of an OWA operator, for conformists and anti-conformists, respectively. Moreover, the assumptions (\ref{eq:4new}) - (\ref{eq:7new}) simply mean that while agents in $N^c$ ($N^a$, respectively) may have different weight
vectors, the number of left and right zeroes is the same for all of them. 
The number of left/right zeroes indicates how many people the agent needs in order to start being influenced towards the yes/no opinion. In particular, a non symmetrical weight vector w.r.t. the number of left and right zeroes means that the agent is biased towards the `yes' or `no' answer, i.e., he needs a different number of people to start being convinced to say `yes' or `no'.

\section{Convergence of the model}\label{sec:convergence}

This section is devoted to the study of absorbing classes in the model introduced in Section \ref{subsec:description}. Unlike the case of a
model with only conformist agents, their study appears to be extremely complex. 

We start by the simple cases where there is no anti-conformist or no conformist
agents. Then we establish the main result (Theorem~\ref{th:main}) giving all
possible absorbing classes in the pure case ($N^m=\emptyset$), and continue with
the mixed case, yielding similar results (Theorem~\ref{th:main2}). These results
are valid without any restriction on the number of agents $n$, nor on any
parameter describing the society. They give an exact description of all possible
absorbing classes (there are 20 classes), with their conditions of
existence. In order to make the results more readable, we then consider the
number of agents tending to infinity and propose some rewriting of the parameters
describing the society. Doing so, it happens that 4 absorbing classes among the
20 ``disappear'', because they exist only as limit cases. Based on that, we
provide a clear analysis of the convergence in three typical types of society.    

Throughout, we will use the following notation: we write
$S\rightarrow T$ if a transition from $S$ to $T$ is possible, i.e.,
$\lambda_{S,T}>0$, and 
$S\srightarrow T$ if $\lambda_{S,T}=1$ (sure transition). We extend the latter
notation to collections of sets: letting $\cS,\cT$ be two nonempty collections
of sets in $2^N$, we write
\[
\cS\srightarrow\cT\quad \Leftrightarrow \quad \forall T\in \cT,\exists S\in\cS
\text{ s.t. }\lambda_{S,T}>0\text{ and } \forall S\in \cS,\forall T\not\in\cT, \lambda_{S,T}=0.
\]

\subsection{Cases with no anti-conformists ($N^a=N^m=\emptyset$) or no conformists ($N^c=N^m=\emptyset$)}\label{sec:noanti}
We mentioned in Section~\ref{sec:anoinf} that for the anonymous model of conformism,
it was found that any state could be absorbing and that other absorbing classes
are intervals of sets. However, the case with $N^a=\emptyset$ is much more particular as every
agent has the same $l,r$, and there are very few possibilities left, as shown
below. To this end, it suffices to rewrite Table~\ref{tab:1} which becomes:
\begin{center}
\begin{tabular}{|c|c|c|}\hline
$0\leq s\leq l^c$ & $l^c<s<n-r^c$ & $n-r^c\leq s\leq n$\\ \hline
$\emptyset$ & $2^N$ & $N$\\ \hline
\end{tabular}
\end{center}
It follows that only $\emptyset,N$ are absorbing states and neither $2^N$ nor
any of its subcollections is an absorbing class, because a transition from any
$S\neq\emptyset,N$ to $\emptyset$ or $N$ is possible.

Both $\emptyset,N$ are possible, regardless of the values of $r^c,l^c$. This
means that the society converges to a consensus, which depends on the initial
state $S$. If $s\leq l^c$, there is convergence to 'no' in one step, and if
$s\geq n-r^c$, there is convergence to 'yes' in one step.


\medskip

The analysis of the extreme case with $N^c=\emptyset$ is also done by rewriting Table~\ref{tab:1}
which becomes now:
\begin{center}
\begin{tabular}{|c|c|c|}\hline
$0\leq s\leq l^a$ & $l^a<s<n-r^a$ & $n-r^a\leq s\leq n$\\ \hline
$N$ & $2^N$ & $\emptyset$\\ \hline
\end{tabular}
\end{center}
Then there is only one absorbing class which is the cycle $\emptyset\srightarrow
N\srightarrow \emptyset$. We see that, without much surprise, a society of
anti-conformist agents can never reach a stable state. 

\subsection{The pure case ($N^m=\emptyset$)}\label{subsec:convergence-pure}

We observe the following basic facts:
\begin{itemize}
\item[(F0)] $\emptyset\srightarrow N^a$, $N\srightarrow N^c$ (as already observed).
\item[(F1)] If $\cS\srightarrow \cT$, $\cS'\srightarrow \cT'$ and $\cS\subset \cS'$,
  then $\cT\subseteq \cT'$.
\item[(F2)]  Applying (F0) and (F1), we find that in a transition $\cS\rightarrow\cT$, $\emptyset\in\cS$ implies
  $N^a\in\cT$ and $N\in\cS$ implies $N^c\in \cT$.
\item[(F3)] Consider $\cS\srightarrow\cT_1\srightarrow\cdots\srightarrow\cT_p$,
  with $p\geq 2$. If $\cS\subseteq \cT_1$, then
  $\cS\subseteq\cT_1\subseteq\cdots\subseteq\cT_p$.
\item[(F4)] $2^N$ is a possible absorbing class. Indeed, take
  $l^c=r^c=l^a=r^a=0$. From Table~\ref{tab:1} we immediately see that for any
  $S\neq\emptyset,N$ we have $S\srightarrow 2^N$. Since the power set of the set
  of states is the ``default'' absorbing class when no other can exist, we
  exclude it from our study and do not consider transitions to $2^N$.
\item[(F5)] From Table~\ref{tab:1}, we see that we have to deal only with the sets
  $\emptyset, N^a,N^c,N$ and the intervals $[\emptyset,N^c], [\emptyset,N^a],
  [N^a,N],[N^c,N]$ ($2^N$ being excluded by (F4)), i.e., only these can be
  constituents of an absorbing class. We put
\[
\BB=\{\{\emptyset\},  \{N^a\},\{N^c\},\{N\},[\emptyset,N^c], [\emptyset,N^a],
                   [N^a,N],[N^c,N]\}
\]
 the set of collections relevant to our study. Intervals not reduced to a
 singleton are called \textit{nontrivial intervals}.
\item[(F6)] $\cS\subseteq 2^N$ is an absorbing class if and only if
  $\cS\srightarrow \cS$ and $\cS$ is \textit{connected} (i.e., there is a path
  (chain of transitions) from $S$ to $T$ for any $S,T\in\cS$).
\end{itemize}
(F6) will be our unique tool to find aperiodic absorbing classes, while periodic
classes are of the form $\cS_1\srightarrow\cdots\srightarrow\cS_p$,
with $\cS_1,\ldots,\cS_p\subset 2^N$ and being pairwise disjoint (no common set
between $\cS_i,\cS_j$), and $\cS_1\cup\cdots\cup\cS_p$ must be connected.

Since $N^m=\emptyset$, we have $n^a=n-n^c$, where $n^a = |N^a|$ and $n^c = |N^c|$. 
Hence, the model is entirely
determined by $l^c,r^c,l^a,r^a,n^c,n$. We recall that these parameters must
satisfy the following constraints:
\begin{align*}
0&\leq l^a+r^a  <n\\
0&\leq l^c+r^c  <n\\
0&<n^c<n.
\end{align*} 

Based on these facts, we can show the main result of this section. 
\begin{Theorem}\label{th:main}
Assume that $N^m =\emptyset $, $N^a \neq \emptyset $ and $N^c \neq
\emptyset$. There are twenty possible absorbing classes which are\footnote{We use the standard notation $\vee$ and $\wedge$ to denote $\max$ and $\min$, respectively.}:
\begin{enumerate}
\item Either one of the following singletons:
  \begin{enumerate}
  \item[(1)] $N^a $ if and only if $n^c\geq (n-l^c)\vee(n-l^a)$;
  \item[(2)]  $N^c $ if and only if $n^c\geq (n-r^c)\vee(n-r^a)$;
  \end{enumerate}
\item or one of the following cycles and periodic classes: 
  \begin{itemize}
  \item[(3)]  $ N^a  \xrightarrow{1} \emptyset \srightarrow N^a $ if and only if
    $n-l^c\leq n^c\leq r^a$;
  \item[(4)]  $N^c \xrightarrow{1} N  \xrightarrow{1} N^c $ if and only if $n-r^c\leq
    n^c\leq l^a$;
  \item[(5)]  $N^a \xrightarrow{1} N^c \xrightarrow{1} N^a $ if and only if $n^c\leq
    l^c\wedge l^a\wedge r^c\wedge r^a$;
  \item[(6)]  $ \emptyset \xrightarrow{1} N^a \xrightarrow{1} N^c \xrightarrow{1}
    \emptyset $ if and only if $n^c\leq r^c\wedge r^a\wedge l^c$ and $n^c\geq
    n-r^a$;
  \item[(7)]  $N^a \xrightarrow{1} N \xrightarrow{1} N^c \xrightarrow{1} N^a $ if and
    only if $n^c\leq l^c\wedge l^a\wedge r^c$ and $n^c\geq n-l^a$;
  \item[(8)]  $ N^a \xrightarrow{1} [\emptyset, N^c] \xrightarrow{1} N^a $ if and only
    if $n^c\leq l^c\wedge l^a\wedge r^a$ and  $r^c<n^c<n-l^c$;
  \item[(9)]  $N^c \xrightarrow{1} [N^a, N] \xrightarrow{1} N^c $ if and only if $n^c\leq r^c\wedge r^a\wedge l^a$ and  $l^c<n^c<n-r^c$;
  \item[(10)]  $[\emptyset, N^c] \xrightarrow{1} [N^a,N] \xrightarrow{1} [\emptyset,
    N^c] $ if and only if $r^c\vee l^c<n^c\leq r^a\wedge l^a\wedge
    (n-l^c-1)\wedge (n-r^c-1)$;
  \end{itemize}			
\item or one of the following intervals or union of intervals:
  \begin{itemize}
  \item[(11)]  $[\emptyset, N^a]$ if and only if $(n-l^c)\vee(r^a+1)\leq n^c<n-l^a$;
  \item[(12)]  $[N^c, N] $ if and only if $(n-r^c)\vee(l^a+1)\leq n^c<n-r^a$;
  \item[(13)]  $ [\emptyset, N^a] \cup [\emptyset, N^c]  $ if and only if
      $l^c\geq n-r^a$ and $n^c\in\big(\left]r^c,n-l^c\right[\cap
    \left]l^a,n-r^c\right[\big)\cup\big(\big(\left]l^a,n-r^a\right[\cup\left]l^c,n-r^c\right[\big)\cap
    \left]0,r^c\right]\big)$;
  \item[(14)]  $ [N^a, N] \cup [N^c,N] $ if and only if $l^a \geq n-r^c $ and $ n^c \in \big(\left]l^c,n-r^c\right[\cap
  	\left]r^a,n-l^c\right[\big)\cup\big(\big(\left]r^a,n-l^a\right[\cup\left]r^c,n-l^c\right[\big)\cap
  	\left]0,l^c\right]\big)  $; 
\item[(15)] $[\emptyset,N^a]\cup\{N^c\}$ if and only if $l^c+r^c=n-1$, $r^a\geq
  r^c$,  $l^c>l^a$ and
  $l^a<n^c< (n-r^a)\wedge(n-l^c)$;
\item[(16)]  $[N^c,N]\cup\{N^a\}$ if and only if $l^c+r^c=n-1$, $l^a\geq l^c$
   $r^c>r^a$ and
    $r^a<n^c< (n-l^a)\wedge(n-r^c)$;
  \item[(17)] $[\emptyset,N^c]\cup\{N^a\}$ if and only if $l^a+r^a=n-1$, $l^c\geq l^a$,
  $n^c<n-r^c$ and $n^c\in\left]r^c,n-l^c\right[\cup\left]l^c,r^c\right]$;
\item[(18)] $[N^a,N]\cup\{N^c\}$ if and only if $l^a+r^a=n-1$, $r^c\geq r^a$,
  $n^c<n-l^c$ and $n^c\in\left]l^c,n-r^c\right[\cup\left]r^c,l^c\right]$.
  \item[(19)]  $[\emptyset,N^a]\cup[N^c,N]$ if and only if $l^c+r^c=n-1$
    and $l^a\vee r^a<n^c\leq l^c\wedge r^c$;
  \item[(20)] $2^N$ otherwise.
  \end{itemize}  
\end{enumerate}
Moreover, if $l^c+r^c=n-1$, then cases (6), (7), (8), (13) and (14) become
impossible while (15) and (16) are specific to this case, and if $l^a+r^a=n-1$,
then cases (11) and (12) become impossible, while (17) and (18) are specific to
this case.
\end{Theorem}
The proof of Theorem \ref{th:main} can be found in the appendix.

\medskip

A complete analysis and interpretation of these results seems to be out of
reach. Nevertheless, this becomes quite possible and instructive when
considering special situations (see below). We recall that these results are
given without any simplifying assumption and are valid whatever the value of the
parameters describing the society. Once they are known, Theorem~\ref{th:main}
immediately gives the possible absorbing classes. We stick in what follows to
general comments, reserving more specific ones and a deeper analysis to the
particular situations studied hereafter:
\begin{enumerate}
\item Although there are many absorbing classes, they can be grouped in three
  categories. The first category ((1) and (2)) comprises absorbing states
  (singleton classes) and represent a stable state of the society in the long
  run, which happens to be a polarization: a part of the society says 'yes' (the
  anti-conformists in case (1), the conformists in case (2)), while the other
  part says 'no'. Observe that consensus never occurs.

The second category (from (11) to (20)) comprises aperiodic classes which are
not reduced to singletons. (11) and (12) are simple cases where the class is an
interval. For (11), it means that in the long run all the conformists say `no',
while the anti-conformists have a chaotic behavior, in the sense that, at each
time step, a (different) fraction of them says 'yes' while the remaining part
says 'no'. We may call this a {\it fuzzy polarization}. Cases (13) through (18)
can be interpreted in a similar way, although the behavior is yet more
chaotic: taking (13) for example, in the long run, at each time step the set of
agents saying 'yes' can be any fraction of anti-conformists or any fraction of
conformists (but not a mixed group). We may call this a {\it chaotic
  polarization}. The extreme case of chaotic behavior is given by (20): at any
time step in the long run, the set of agents saying 'yes' can be any subset of
agents.  

The third category (from (3) to (10)) comprises cycles and periodic classes,
which also have a natural interpretation. Cycles (from (3) to (7)) are
succession in time of states of polarization and/or consensus, endlessly in the
same order. Periodic classes ((8) to (10)) mix the cyclic behavior with
intervals as in the case of fuzzy polarization, and for this reason may be
called {\it fuzzy cycles}. For example, in case (8) at some
step in the long run all anti-conformists say `yes', in the following step they
all say `no' but a fraction of conformists says `yes', then in the next
step again all anti-conformists say `yes', etc.
\item Cases (1) to (20) are not exclusive. This can immediately be seen by
  considering cases (1) and (2). Indeed, for a given society, which is
  represented by the set of parameters $n$, $n^c$, $l^a$, $r^a$, $l^c$ and $r^c
  $, under some conditions both cases (1) and (2) are possible, and therefore
  two different absorbing classes might occur, $N^a$ and $N^c$. However, the
  process will end up in only one of them, with some probability.
\item The analysis for conformists and anti-conformists is not symmetric. For
  example, $[\emptyset,N^a]$ is a possible absorbing class but not
  $[\emptyset,N^c]$. However, while there is no symmetry between ``$a$'' and
  ``$c$'' in this framework, there exists symmetry between $S$ and $N\setminus
  S$ as pointed out in Lemma \ref{lem:sym}.
\end{enumerate}

\subsection{The case of only mixed agents ($N^m=N$)}
Prior to the study of the mixed case, we consider the situation where the
society is formed only by mixed agents: $N=N^m$. It is easily done by using
Eqs. (\ref{eq:7}) and (\ref{eq:8}):
\begin{align*}
p_i(s) = 0 & \Leftrightarrow n-r^a\leq s\leq l^c\\
p_i(s) = 1 & \Leftrightarrow n-r^c\leq s\leq l^a
\end{align*}
which permits to rewrite Table~\ref{tab:1}:
\begin{center}
\begin{tabular}{|c|c|c|c|}\hline
 & $0\leq s<n-r^c$ & $n-r^c\leq s\leq l^a$ &    $l^a< s\leq n$\\ \hline
$0\leq s< n-r^a$ & $2^N$ & $N $& $2^N$\\ \hline
$n-r^a\leq s\leq l^c$ & $\emptyset$ & does not occur & $\emptyset$\\ \hline
$l^c<s\leq n$ & $2^N$ & $N$& $2^N$\\ \hline
\end{tabular}
\end{center}
We can see that $\emptyset,N$ are not absorbing states, hence the only absorbing
class is $2^N$.

\subsection{The (general) mixed case ($N^m \neq \emptyset$)}\label{subsec:convergence-mixed}
Next, we consider a society in which pure conformists, pure anti-conformists, and mixed agents co-exist, i.e., $n=n^a + n^c + n^m$, with $n^a >0$, $n^c > 0$ and $n^m > 0$. We get the following result.

\begin{Theorem}\label{th:main2}
Assume that $N^m \neq \emptyset $, $N^a \neq \emptyset $ and $N^c \neq
\emptyset$. Let $\overline{N}^{a}= N^a\cup N^m$ and $\overline{N}^{c}= N^c\cup
N^m$. There are twenty possible absorbing classes which are:
\begin{enumerate}
\item Either one of the following intervals:
  \begin{enumerate}
	\item[(1)] $[N^a,\overline{N}^{a}] $ if and only if $n^c\geq (n-l^c)\vee(n-l^a)$;	
	\item[(2)] $[N^c, \overline{N}^{c}] $ if and only if $n^c\geq (n-r^c)\vee(n-r^a)$;
	\item[(3)] $[\emptyset, \overline{N}^{a}]$ if and only if $(n-l^c)\vee(r^a+1)\leq n^c<n-l^a$;
	\item[(4)] $[N^c, N] $ if and only if $(n-r^c)\vee(l^a+1)\leq n^c<n-r^a$;
  \end{enumerate}
\item or one of the following periodic classes:
\begin{enumerate}
	\item[(5)] $ [N^a,\overline{N}^{a}] \xrightarrow{1} \emptyset \srightarrow [N^a,\overline{N}^{a}] $ if and only $n-l^c\leq n^c\leq r^a-n^m$;
	\item[(6)] $[N^c, \overline{N}^{c}] \xrightarrow{1} N \xrightarrow{1} [N^c,\overline{N}^{c} ] $ if and only if $n-r^c\leq n^c\leq l^a-n^m$;
	\item[(7)] $[N^a,\overline{N}^{a}] \xrightarrow{1} [N^c, \overline{N}^{c}] \xrightarrow{1} [N^a,\overline{N}^{a}] $ if and only if $n^c+n^m\leq l^c\wedge l^a\wedge r^c\wedge r^a$;
	\item[(8)] $ \emptyset \xrightarrow{1} [N^a,\overline{N}^{a}] \xrightarrow{1} [N^c,\overline{N}^{c}] \xrightarrow{1} \emptyset $ if and only if $n^c+n^m\leq r^c\wedge r^a\wedge l^c$ and $n^c\geq n-r^a$;
	\item[(9)] $[N^a,\overline{N}^{a}] \xrightarrow{1} N\xrightarrow{1} [N^c,\overline{N}^{c}] \xrightarrow{1} [N^a ,\overline{N}^{a}] $ if and only if $n^c+n^m\leq l^c\wedge l^a\wedge r^c $ and $n^c\geq n-l^a$;
	\item[(10)] $ [N^a,\overline{N}^{a}] \xrightarrow{1} [\emptyset, \overline{N}^{c}] \xrightarrow{1} [N^a,\overline{N}^{a}] $ if and only if $n^c+n^m\leq l^c \wedge l^a\wedge r^a$ and $r^c-n^m<n^c<n-l^c$;
	\item[(11)] $[N^c,\overline{N}^{c}] \xrightarrow{1} [N^a, N] \xrightarrow{1} [N^c,\overline{N}^{c}] $ if and only if $n^c+n^m\leq r^c\wedge r^a\wedge l^a$ and $l^c-n^m<n^c<n-r^c$;
	\item[(12)] $[\emptyset, \overline{N}^{c}] \xrightarrow{1} [N^a,N] \xrightarrow{1} [\emptyset,\overline{N}^{c}] $ if and only if $ l^c \vee r^c <n^c+n^m \leq r^a \wedge l^a \wedge (n-(l^c+1))\wedge (n-(r^c+1)) $;
\end{enumerate}
\item or one of the following unions of intervals:
		\begin{enumerate}		
	\item[(13)] $[\emptyset,\overline{N}^{a}]\cup [\emptyset,\overline{N}^{c}] $ if and only if $l^c\geq n-r^a $ and 	
	$$n^c\in\big(\left]r^c-n^m,n-l^c\right[\cap \left]l^a,n-r^c\right[\big)\cup \big(\big(\left]l^a-n^m,n-r^a\right[\cup\left]l^c-n^m,n-r^c\right[\big) \cap \left]0,r^c-n^m\right]\big);$$
	\item[(14)] $[N^a,N]\cup [N^c,N] $ if and only if $r^c \geq n-l^a $ and 
	$$n^c\in\big(\left]l^c-n^m,n-r^c\right[\cap \left]r^a,n-l^c\right[\big)\cup \big(\big(\left]r^a-n^m,n-l^a\right[\cup\left]r^c-n^m,n-l^c\right[\big)\cap \left]0,l^c-n^m\right]\big);$$
	\item[(15)]  $[\emptyset,\overline{N}^{a}] \cup [N^c,\overline{N}^{c}] $
          if and only if $l^c+r^c=n-1$, $r^c\leq r^a$, $l^c>l^a $, $n^c < n-l^c $ and $ l^a <n^c+n^m<n-r^a $;
	\item[(16)]  $[N^c,N] \cup [N^a,\overline{N}^{a}] $ if and only if $l^c+r^c=n-1$, $l^c\leq l^a$, $r^a<r^c $, $n^c<n-r^c $ and $ r^a <n^c+n^m<n-l^a $;
	\item[(17)]  $ [\emptyset, \overline{N}^{c}] \cup [N^a,\overline{N}^{a}] $ if and only if $l^a+r^a= n-1$, $l^c\geq l^a$, $ n^c<n-r^c $ and $ n^c \in \left] r^c-n^m,n-l^c\right[  \cup \left] l^c-n^m,r^c-n^m \right[ $;
	\item[(18)]  $[N^a,N] \cup [N^c,\overline{N}^{c}] $ if and only if
          $l^a+r^a= n-1$, $r^c\geq r^a$, $ n^c<n-l^c $ and $ n^c \in \left]
          l^c-n^m,n-r^c\right[  \cup \left] r^c-n^m,l^c-n^m \right[ $;
          \item [(19)]  $[\emptyset, \overline{N}^a] \cup [N^c,N] $ if and only if  $ l^c+r^c=n-1 $ and $ l^a \vee r^a <n^c \leq l^c \wedge r^c $;
\end{enumerate}
\item[(20)] $2^N $ otherwise.
\end{enumerate}
\end{Theorem}

The proof is similar to the one of Theorem \ref{th:main} and is omitted here, but is available upon request. 

Theorems \ref{th:main} and \ref{th:main2} lead to clear conclusions concerning
the comparison of absorbing classes in the pure and mixed cases. First of all,
when mixed agents exist in a society, a polarization into two groups (one saying
`yes' and another one saying `no', which was the case under $N^m = \emptyset$)
is not possible anymore. However, under the same conditions as before (see cases
(1) and (2) in Theorem \ref{th:main}), $N^a$ and $N^c$ are now replaced by
$[N^a,\overline{N}^{a}]$ and $[N^c, \overline{N}^{c}]$. In other words, while
anti-conformists (conformists, respectively) continue saying `yes' and
conformists (anti-conformists, respectively) say `no' forever, the new type of
individuals -- mixed agents -- oscillate between `yes' and `no'. In the pure
case, two (simple) intervals $[\emptyset, N^a]$ and $[N^c, N] $ (cases (11) and
(12) in Theorem \ref{th:main}) are possible absorbing classes. With the presence
of mixed agents, we have the corresponding intervals $[\emptyset,
  \overline{N}^{a}]$ and $[N^c, N] $ (cases (3) and (4) in Theorem
\ref{th:main2}) under the same conditions as in the pure case. This means that
while conformists do not change their behavior when mixed agents join the
society and say either `no' (absorbing class $[\emptyset, \overline{N}^{a}]$) or
`yes' (absorbing class $[N^c, N] $) forever, now besides anti-conformists also
mixed agents oscillate. Another consequence of the presence of mixed agents on
possible absorbing classes is that cycles (i.e., periodic classes with only
single states, cases (3) through (7) in Theorem \ref{th:main}) are not
possible anymore. Instead, we have eight periodic classes with mixed agents
oscillating (cases (5) till (12) in Theorem \ref{th:main2}) that correspond to
absorbing classes (3) - (10) of Theorem \ref{th:main}. The conditions for the
existence of these periodic classes in the mixed case are the same as the ones
for the corresponding `pure' cases, but adjusted by the presence of $n^m$ mixed
agents. Finally, the unions of intervals in the mixed case (absorbing classes
(13) till (18) in Theorem \ref{th:main2}) correspond to the unions of intervals
(13) till (18) in the pure case (Theorem \ref{th:main}), but again with mixed
agents oscillating and the conditions taking into account $n^m$.

\subsection{Analysis of the pure case when $n$ tends to infinity}
We make the assumption that the number of agents is very large and approximate
this situation by making $n$ tend to infinity. For notational convenience, each
of the previous parameters $n^a,l^a,r^a,l^c,r^c$ is now divided by $n$, keeping
(with some abuse) the same notation for these parameters, so that
now these are real numbers in $[0,1]$. It follows that
\begin{align*}
n^c &= 1-n^a\\
l^a+r^a & <1\\
l^c+r^c & <1,
\end{align*}
Note that the particular cases $l^a+r^a=n-1$, $l^c+r^c=n-1$ appearing in classes
(15) to (19) become limit cases $l^a+r^a\rightarrow 1$, $l^c+r^c\rightarrow 1$,
making the latter classes appearing only as limit cases. 

We study in details in the rest of this section some specific
situations. Observe that from the results of Section~\ref{sec:noanti}, the
``model'' is not continuous in $n^a$ at 0, in the sense that if $n^a>0$, we have
proved that $2^N$ can be an absorbing class (it suffices to take
$l^c=r^c=l^a=r^a=0$). Similarly, it is not continuous in $n^a$ at 1.

Given an aggregation rule with $l,r$ specified, we introduce the new
parameter 
\[
\gamma = \frac{1}{1-r-l},
\]
which is the average slope of the function giving the probability $p$ to say
'yes' given the number of agents saying `yes'. $l$ indicates how much an agent
needs agents saying `yes' before starting to change his mind (this could be
called the {\it firing threshold}), while $\gamma$ measures the {\it
  reactiveness} once he has started to change his mind. Note that
$\gamma\in\left[\frac{1}{1-l},\infty\right[$. On the other hand, $1-r$ is the
    {\it saturation threshold}, beyond which there is no more change of opinion
    for the agent. In a dual way, $r$ can be interpreted as how much an agent
    needs agents saying `no' before starting to change his mind, while $1-l$ is
    the saturation threshold beyond which additional agents saying `no' have no
    effect.  

The pair of parameters $(l,\gamma)$ may be easier to interpret than the pair
$(l,r)$, and so we will use it in the sequel whenever convenient. We have
$r=1-l-\frac{1}{\gamma}$.

In what follows, we make a detailed analysis and interpretation of the
convergence in several typical situations. We assume throughout $N^m=\emptyset$.

\paragraph{Situation 1: $l^a=l^c=l$ and $r^a=r^c=r$.} This depicts a society
where all agents, conformists or anti-conformists, have the same influenceability
characteristics. Among the initially 15 possible absorbing classes, only the
following ones remain possible:
\begin{itemize}
\item $N^a$ if and only if $n^a\leq l$
\item $N^c$ if and only if $n^a\leq r$
\item cycle $N^a\srightarrow N^c\srightarrow N^a$ if and only if $n^a\geq 1-l$ and $n^a\geq 1-r$
\item $2^N$ otherwise.
\end{itemize}
Let us translate this with the pair $(l,\gamma)$.  We obtain:
\begin{enumerate}
\item $N^a$ if and only if $n^a\leq l$
\item $N^c$ if and only if $n^a\leq 1-l-\frac{1}{\gamma}$
\item cycle $N^a\srightarrow N^c\srightarrow N^a$ if and only if $n^a\geq 1-l$ and $n^a\geq \frac{1}{\gamma}+l$
\item $2^N$ otherwise.
\end{enumerate}
We make a ``phase diagram'' with the three parameters $n^a,l,\gamma$ showing the
possible absorbing classes, keeping in mind that $\gamma\geq \frac{1}{1-l}$.
\begin{figure}[htb]
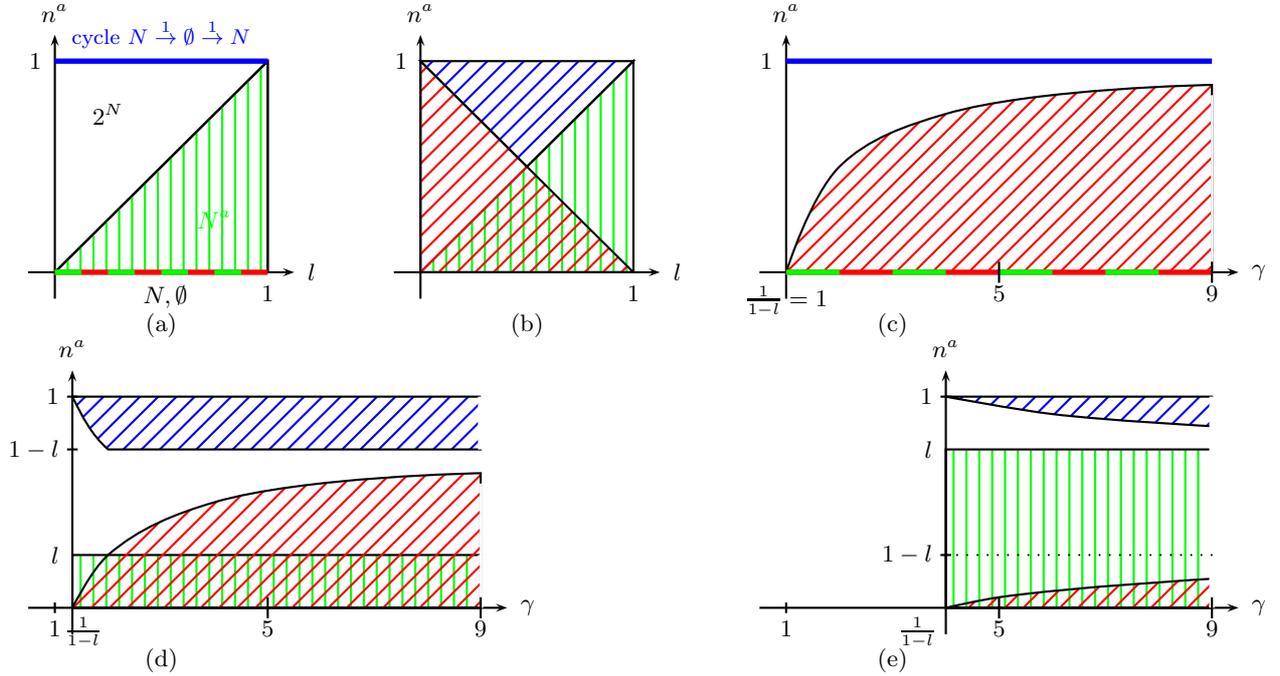

\begin{center}
\psset{unit=0.7cm}
\pspicture(-0.5,-1)(4.5,5)
\psline{->}(-0.5,0)(4.5,0)
\psline{->}(0,-0.5)(0,4.5)
\uput[0](4.5,0){$l$}
\uput[-90](4,0){\small $1$}
\uput[90](0,4.5){$n^a$}
\uput[180](0,4){\small $1$}
\rput(2,-1){(a)}
\psline(0,4)(4,4)(4,0)
\psline(0,0)(4,4)
\rput(3,1){\green $N^a$}
\rput(1,3){$2^N$}
\pspolygon[hatchcolor=green,hatchwidth=0.8pt,hatchangle=0,fillstyle=vlines](0,0)(4,4)(4,0)
\psline[linewidth=2pt,linecolor=red](0,0)(4,0)
\psline[linewidth=2pt,linecolor=green](0,0)(0.5,0)
\psline[linewidth=2pt,linecolor=green](1,0)(1.5,0)
\psline[linewidth=2pt,linecolor=green](2,0)(2.5,0)
\psline[linewidth=2pt,linecolor=green](3,0)(3.5,0)
\psline[linewidth=2pt,linecolor=blue](0,4)(4,4)
\uput[90](2,4){\blue\scriptsize cycle $N\srightarrow \emptyset\srightarrow N$}
\uput[-90](2,0){ $N,\emptyset$}
\endpspicture
\hfill
\pspicture(-0.5,-1)(4.5,5)
\psline{->}(-0.5,0)(4.5,0)
\psline{->}(0,-0.5)(0,4.5)
\uput[0](4.5,0){$l$}
\uput[-90](4,0){\small $1$}
\uput[90](0,4.5){$n^a$}
\uput[180](0,4){\small $1$}
\rput(2,-1){(b)}
\pspolygon[hatchcolor=green,hatchwidth=0.8pt,hatchangle=0,fillstyle=vlines](0,0)(4,4)(4,0)
\pspolygon[hatchcolor=red,hatchwidth=0.8pt,fillstyle=hlines](0,0)(0,4)(4,0)
\pspolygon[hatchcolor=blue,hatchwidth=0.8pt,fillstyle=hlines](0,4)(2,2)(4,4)
\endpspicture
\hfill
\pspicture(-0.5,-1)(8.5,5)
\psline{->}(-0.5,0)(8.5,0)
\psline{->}(0,-0.5)(0,4.5)
\uput[0](8.5,0){$\gamma$}
\uput[-90](4,0){\small $5$}
\uput[-90](8,0){\small $9$}
\uput[-90](0,0){\small $\frac{1}{1-l}=1$}
\uput[90](0,4.5){$n^a$}
\uput[180](0,4){\small $1$}
\rput(2,-1){(c)}
\pscustom[hatchcolor=red,hatchwidth=0.8pt,fillstyle=hlines]{
\pscurve(0,0)(1,2)(3,3)(8,3.55)
\psline(8,0)(0,0)
}
\psline[linecolor=white](8,3.355)(8,0)
\psline[linewidth=2pt,linecolor=red](0,0)(8,0)
\psline[linewidth=2pt,linecolor=green](0,0)(1,0)
\psline[linewidth=2pt,linecolor=green](2,0)(3,0)
\psline[linewidth=2pt,linecolor=green](4,0)(5,0)
\psline[linewidth=2pt,linecolor=green](6,0)(7,0)
\psline[linewidth=2pt,linecolor=blue](0,4)(8,4)
\psline(4,-0.2)(4,0.2)
\psline(8,-0.2)(8,0.2)
\endpspicture

\medskip

\pspicture(-0.5,-1)(8.5,5)
\psline{->}(-0.5,0)(8.5,0)
\psline{->}(0.33,-0.5)(0.33,4.5)
\uput[0](8.5,0){$\gamma$}
\uput[-90](4,0){\small $5$}
\uput[-90](8,0){\small $9$}
\uput[-45](0,0){\small $\frac{1}{1-l}$}
\uput[90](0.33,4.5){$n^a$}
\uput[180](0.33,4){\small $1$}
\rput(2,-1){(d)}
\uput[180](0.33,1){$l$}
\uput[-90](0,0){\small $1$}
\psline(0,-0.1)(0,0.1)
\pspolygon[hatchcolor=green,hatchwidth=0.8pt,hatchangle=0,fillstyle=vlines](0.33,0)(0.33,1)(8,1)(8,0)
\psline[linecolor=white,linewidth=2pt](8,1)(8,0)
\pscustom[hatchcolor=red,hatchwidth=0.8pt,fillstyle=hlines]{
\pscurve(0.33,0)(1,1)(3,2)(8,2.55)
\psline(8,0)(0.33,0)
}
\psline[linecolor=white](8,2.355)(8,0)
\pscustom[hatchcolor=blue,hatchwidth=0.8pt,fillstyle=hlines]{
\psline(8,4)(0.33,4)
\psecurve(0,5)(0.33,4)(1,3)(3,2)
\psline(1,3)(8,3)(8,4)
}
\psline[linecolor=white,linewidth=2pt](8,3)(8,4)
\psline(4,-0.2)(4,0.2)
\psline(8,-0.2)(8,0.2)
\psline(0.23,4)(0.43,4)
\psline(0.23,3)(0.43,3)
\uput[180](0.33,3){$1-l$}
\endpspicture
\hfill
\pspicture(-0.5,-1)(8.5,5)
\psline{->}(-0.5,0)(8.5,0)
\psline{->}(3,-0.5)(3,4.5)
\uput[0](8.5,0){$\gamma$}
\uput[-90](4,0){\small $5$}
\uput[-90](8,0){\small $9$}
\uput[-135](3,0){\small $\frac{1}{1-l}$}
\uput[90](3,4.5){$n^a$}
\uput[180](3,4){\small $1$}
\uput[-90](0,0){\small $1$}
\psline(0,-0.1)(0,0.1)
\psline(2.9,4)(3.1,4)
\psline(2.9,1)(3.1,1)
\rput(2,-1){(e)}
\uput[180](3,1){$1-l$}
\pspolygon[hatchcolor=green,hatchwidth=0.8pt,hatchangle=0,fillstyle=vlines](3,0)(3,3)(8,3)(8,0)
\pscustom[hatchcolor=red,hatchwidth=0.8pt,fillstyle=hlines]{
\pscurve(3,0)(4,0.2)(8,0.55)
\psline(8,0)(3,0)
}
\pscustom[hatchcolor=blue,hatchwidth=0.8pt,fillstyle=hlines]{
\pscurve(3,4)(5,3.67)(8,3.44)
\psline(8,4)(3,4)
}
\psline[linecolor=white,linewidth=2pt](8,4)(8,0)
\psline(4,-0.2)(4,0.2)
\psline(8,-0.2)(8,0.2)
\psline[linestyle=dotted](3,1)(8,1)
\uput[180](3,3){$l$}
\endpspicture

\end{center}
\caption{Phase diagram for Situation 1: (a) $\gamma$ has minimum value
  $\frac{1}{1-l}$ ($r=0$); (b)
  $\gamma\rightarrow\infty$; (c) $l=0$; (d) $l\in[0,1/2]$; (e) $l\in\left[1/2,1\right[$. 
Color code:
      white=$2^N$, blue=cycle, red=$N^c$, green=$N^a$}
      \label{fig:1}
\end{figure}

We comment on these phase diagrams.
\begin{itemize}
\item The two first cases (i) and (ii) of absorbing classes are
  ``polarization'', the last case (iv) is ``chaos'' (no
  convergence). The extreme cases ($n^a=0$ or 1) are already commented (see
  Section~\ref{sec:noanti}). Also, it can be checked that when
  $\gamma$ tends to infinity, the limit cases (15) to (19) do not appear as the
  existence conditions become contradictory.
\item When $n^a$ increases from 0 to 1, we go from consensus, next to polarization,
  next to chaos, and finally to a cycle.
\item When the firing threshold $l$ is very low (c), there is a cascade effect
  leading to a polarization where all conformist agents say `yes', which increases
  with reactiveness. Indeed, suppose that all agents say `no'. Then, as $l$ is
  very low, all anti-conformists start to say `yes', which make gradually the
  conformists saying `yes'. If $n^a$ is not too large, the conformists rapidly
  reach the consensus `yes'. Otherwise, as anti-conformists will say `no' again,
  the non-negligible proportion of `no' causes trouble in the convergence and a
  chaotic situation may appear. As the reactiveness increases, the chaotic
  behavior is less and less probable.
\item Similarly, when the firing threshold is high (e), there is a cascade
  effect leading all conformists to say `no', if the proportion of
  anti-conformists is not too small but less than the firing threshold. Indeed,
  suppose that all agents say `yes' at some time. Then all anti-conformists will
  say `no'. As the firing threshold is high, some conformists will start to say
  `no', and there will be more and more. At the same time, as the number of
  `yes' in the society is decreasing, the anti-conformists will gradually change
  to `yes'. The situation of polarization remains stable unless the number of
  anti-conformists exceeds the firing threshold, in which case a chaotic
  situation (or even a cycle) occurs. 
\item (d) shows an intermediary situation where both cascades can occur. The
  higher the firing threshold, the higher the probability to have a cascade of
  `no' among the conformists. The two cases (c) and (e) show how, in a society
  of conformists, the opinion can be manipulated by introducing a relatively
  small proportion of anti-conformists. The final opinion depends
  essentially on the firing threshold.  
\end{itemize}

\paragraph{Situation 2: $l^a=r^a$ and $l^c=r^c$.} Here, there is a symmetry
between $l$ and $r$. As mentioned before, this means that agents treat in the
same way `yes' and `no' opinions. This assumption might be relevant for instance
when voting for two candidates. However, it might not be relevant when saying
`yes' means `adopting a new technology', where a bias towards a status quo or a
bias towards technology adoption makes sense. Under this assumption, the
possible absorbing classes are:
\begin{itemize}
\item $N^a,N^c$ iff $n^a\leq l^a$ and $n^a\leq l^c$ (referred hereafter as ``polarization'')
\item cycle $N^a\srightarrow N^c\srightarrow N^a$ iff $n^a\geq 1-l^a$ and
  $n^a\geq 1-l^c$ (referred hereafter as ``cycle'')
\item periodic class
  $[\emptyset,N^c]\srightarrow[N^a,N]\srightarrow[\emptyset,N^c]$ iff $n^a\geq
  1-l^a$ and $l^c<n^a<1-l^c$ (referred hereafter as ``fuzzy cycle'')
\item $[\emptyset,N^a],[N^c,N]$ iff $n^a\leq l^c$ and $l^a<n^a<1-l^a$ (referred
  hereafter as ``fuzzy polarization'')
\item $2^N$ (referred hereafter as ``chaos'') otherwise.
\end{itemize}
It can be checked that in the limiting case where $\gamma^a,\gamma^c$ tend to
infinity, classes (15) to (19) are not possible since then $l^a=l^c=1/2$ which
makes the conditions of existence contradictory.

Figure~\ref{fig:2} gives four cuts of the phase diagram with the three parameters
$n^a,l^a,l^c$. Recall that here $l^a,l^c$ vary in $\left[0,1/2\right[$, and
    $\gamma^a=\frac{1}{1-2l^a}$, $\gamma^c=\frac{1}{1-2l^c}$.   
\begin{figure}[htb]
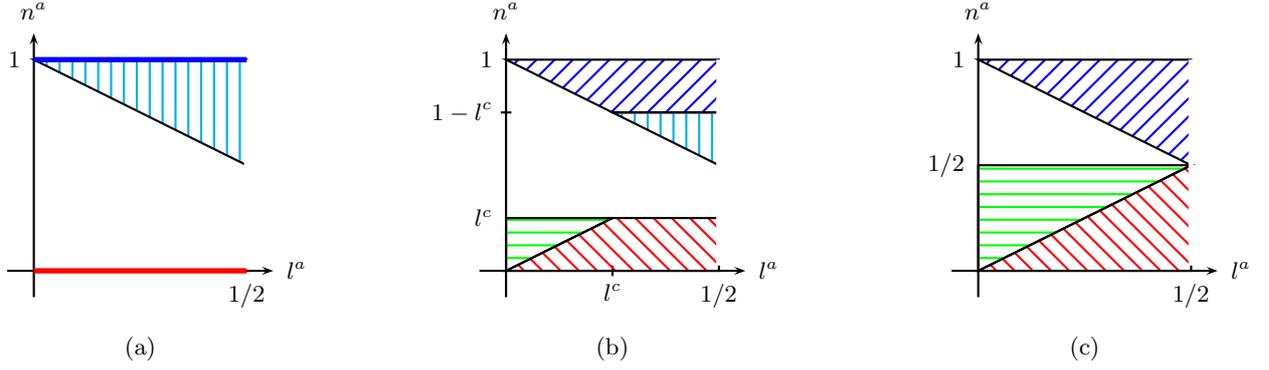

\begin{center}
\psset{unit=0.7cm}
\pspicture(-0.5,-1)(4.5,5)
\psline{->}(-0.5,0)(4.5,0)
\psline{->}(0,-0.5)(0,4.5)
\pspolygon[hatchcolor=cyan,hatchwidth=0.8pt,hatchangle=0,fillstyle=vlines](0,4)(4,2)(4,4)
\psline[linewidth=2pt,linecolor=white](4,0.1)(4,4)
\uput[0](4.5,0){$l^a$}
\uput[-90](4,0){\small $1/2$}
\uput[90](0,4.5){$n^a$}
\uput[180](0,4){\small $1$}
\uput[-90](2,-1){(a)}
\psline[linewidth=2pt,linecolor=red](0,0)(4,0)
\psline[linewidth=2pt,linecolor=blue](0,4)(4,4)
\endpspicture
\hfill
\pspicture(-0.5,-1)(4.5,5)
\psline{->}(-0.5,0)(4.5,0)
\psline{->}(0,-0.5)(0,4.5)
\uput[0](4.5,0){$l^a$}
\uput[-90](4,0){\small $1/2$}
\uput[90](0,4.5){$n^a$}
\uput[180](0,4){\small $1$}
\uput[180](0,3){\small $1-l^c$}
\uput[180](0,1){\small $l^c$}
\uput[-90](2,0){\small $l^c$}
\psline(-0.1,3)(0.1,3)
\psline(2,-0.1)(2,0.1)
\uput[-90](2,-1){(b)}
\pspolygon[hatchcolor=blue,hatchwidth=0.8pt,fillstyle=hlines](0,4)(4,4)(4,3)(2,3)
\pspolygon[hatchcolor=cyan,hatchwidth=0.8pt,hatchangle=0,fillstyle=vlines](2,3)(4,3)(4,2)
\pspolygon[hatchcolor=red,hatchwidth=0.8pt,fillstyle=vlines](0,0)(2,1)(4,1)(4,0)
\pspolygon[hatchcolor=green,hatchwidth=0.8pt,hatchangle=0,fillstyle=hlines](0,0)(0,1)(2,1)
\psline[linewidth=2pt,linecolor=white](4,0.1)(4,4)
\endpspicture
\hfill
\pspicture(-0.5,-1)(4.5,5)
\psline{->}(-0.5,0)(4.5,0)
\psline{->}(0,-0.5)(0,4.5)
\uput[0](4.5,0){$l^a$}
\uput[-90](4,0){\small $1/2$}
\uput[90](0,4.5){$n^a$}
\uput[180](0,4){\small $1$}
\uput[180](0,2){\small $1/2$}
\uput[-90](2,-1){(c)}
\pspolygon[hatchcolor=blue,hatchwidth=0.8pt,fillstyle=hlines](0,4)(4,4)(4,2)
\pspolygon[hatchcolor=red,hatchwidth=0.8pt,fillstyle=vlines](0,0)(4,2)(4,0)
\pspolygon[hatchcolor=green,hatchwidth=0.8pt,hatchangle=0,fillstyle=hlines](0,0)(4,2)(0,2)
\psline[linewidth=2pt,linecolor=white](4,0.1)(4,4)
\endpspicture
\end{center}
\caption{Phase diagram for Situation 2: (a) $l^c=0$; (b)
  $l^c\in\left]0,1/2\right[$; (c)  $l^c\rightarrow 1/2$. Color code:
      white=chaos, blue=cycle, cyan=fuzzy cycle, red=polarization, green=fuzzy polarization}
      \label{fig:2}
\end{figure}
Note that the polarization at $n^a=0$ becomes a consensus (either $N$ or
$\emptyset$). As before, the cycle at $n^a=1$ is $N\srightarrow
\emptyset\srightarrow N$. 

Some comments about these phase diagrams:
\begin{itemize}
\item Compared to Situation~1, the chaos case takes a relatively large area,
  which grows as $l^c$ or $l^a$ tend to 0 (agents have a low firing threshold,
  but a low reactiveness). In particular, it can be observed that when
  conformist agents have a low reactivity, a very small proportion of
  anti-conformists in the society suffices to make it chaotic. 
\item Contrarily to Situation 1, there is no cascade effect. Indeed, the
  absorbing states $N^a$ and $N^c$ always appear together, hence both are
  possible with some probability, or both are impossible. This polarization
  effect happens if the anti-conformists are not ``seen'' by the conformists
  (their number stay below the firing threshold), and all the more since the
  anti-conformists are reactive. Less reactive anti-conformists have a tendency
  to provoke fuzzy polarization. 
\item As for Situation 1, cycles and fuzzy cycles happen all the more since the
  number of anti-conformists is growing. A limit phenomenon happens when
  $l^a,l^c,n^a$ tend all together to 1/2: a kind of ``triple point''
  appears (see (c)), in the sense that the three types of behavior
  (polarization, fuzzy polarization and cycle) happen together, which is also
  visible for Situation~1 (Figure~\ref{fig:1}(b)). Observe that the mix of
  polarization and fuzzy polarization are nothing else than the limit classes
  (15) and (16). According to Theorem~\ref{th:main}, they happen iff
  $l^c,l^a\rightarrow 1/2$ and $n^a=1/2$, which is exactly the locus of this
  triple point.  
\end{itemize}

\paragraph{Situation 3: The case where $n^a$ tends to 0.}
Let us put $n^a=\epsilon>0$, arbitrarily small. Therefore,
$n^c=1-\epsilon$. This case is the most plausible in real situations, as
anti-conformists can be reasonably thought of forming a tiny part of the
society. The crucial question is however to know whether this tiny part can have
a non-negligible effect on the opinion of the society.

The first task is to see which of the 15 classes remain possible. One
can check that:
\begin{itemize}
\item (1) $N^a$ iff  $l^c\wedge l^a\geq \epsilon$;
\item (2) $N^c$ iff  $r^c\wedge r^a\geq\epsilon$;
\item (3) $N^a\srightarrow\emptyset\srightarrow N^a$ iff $l^c\geq\epsilon$ and
  $r^a\geq 1-\epsilon$;
\item (4) $N^c\srightarrow N\srightarrow N^c$ iff $r^c\geq\epsilon$ and
  $l^a\geq 1-\epsilon$
\item Classes (5) to (10) are impossible;
\item (11) $[\emptyset,N^a]$ iff $l^a<\epsilon$, $l^c\geq\epsilon$ and
  $r^a<1-\epsilon$;
\item  (12) $[N^c,N]$ iff $r^a<\epsilon$, $r^c\geq\epsilon$ and
  $l^a<1-\epsilon$;
\item (13) $[\emptyset,N^a]\cup[\emptyset,N^c]$ iff $l^c<\epsilon$,
  $r^c<\epsilon$ and $r^a>1-\epsilon$;
\item (14)  $[N^a,N]\cup[N^c,N]$ iff $l^c<\epsilon$,
  $r^c<\epsilon$ and $l^a>1-\epsilon$;
\item (20) $2^N$ otherwise. 
\end{itemize}
Keeping in mind that $\epsilon$ is small, we can provide the following
interpretation of the above absorbing classes: (1) and (2) are consensus to `no'
and `yes', respectively, up to the negligible fraction of anti-conformists. (3)
is almost the same as (1), while (4) is almost the same as (2). Also, (11) and
(12) are almost the same as (1) and (2), respectively. (13) is a chaotic
situation with mainly a tendency to `no' for the society, while (14) is also a
chaotic situation, but with a tendency of `yes'.

From this analysis, we can draw the following conclusions:
\begin{itemize}
\item Suppose that the conformists have $l^c,r^c>\epsilon$: this means that they
  cannot ``see'' the anti-conformists. Then (1), (2) are  together possible
    as soon as $r^a,l^a>\epsilon$ (the anti-conformists do not react to
    themselves). On the border area where $l^a$ or $r^a$ is smaller than
    $\epsilon$, classes (3) and (11) (almost consensus `no') or classes (4) and
    (12) (almost consensus `yes') appear. 
The situation is made clear by looking at Figure~\ref{fig:3} (recall that $l^a+r^a<1$).
\begin{figure}[htb]
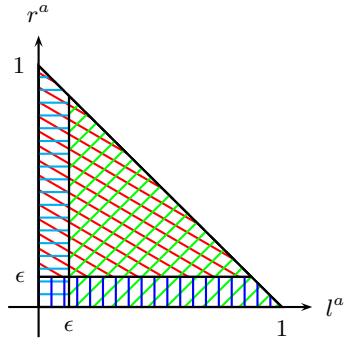

\begin{center}
\psset{unit=0.8cm}
\pspicture(-0.5,-1)(4.5,5)
\psline{->}(-0.5,0)(4.5,0)
\psline{->}(0,-0.5)(0,4.5)
\uput[0](4.5,0){$l^a$}
\uput[-90](4,0){\small $1$}
\uput[-90](0.5,0){\small $\epsilon$}
\uput[90](0,4.5){$r^a$}
\uput[180](0,4){\small $1$}
\uput[180](0,0.5){\small $\epsilon$}
\psline[linestyle=dashed](0,4)(4,0)
\pspolygon[hatchcolor=red,hatchwidth=0.8pt,hatchangle=60,fillstyle=vlines](0,0.5)(0,4)(3.5,0.5)
\pspolygon[hatchcolor=green,hatchwidth=0.8pt,fillstyle=hlines](0.5,0)(0.5,3.5)(4,0)
\pspolygon[hatchcolor=cyan,hatchwidth=0.8pt,hatchangle=0,fillstyle=hlines](0,0)(0,4)(0.5,3.5)(0.5,0)
\pspolygon[hatchcolor=blue,hatchwidth=0.8pt,hatchangle=0,fillstyle=vlines](0,0)(4,0)(3.5,0.5)(0,0.5)
\endpspicture
\caption{Phase diagram for Situation 3, with $l^c,r^c>\epsilon$. Color code:
  green=$N^a$, red=$N^c$, cyan=almost consensus `no'
((3) or (11)), blue=almost consensus `yes' (4) or (12)}
\label{fig:3}
\end{center}
\end{figure}
Observe that in all parts of the triangle, both consensus `yes' and `no'
  coexist. Therefore, no  cascades of `yes' or `no' may occur. Also, no
cycle nor chaotic behavior is possible, and we conclude that this situation is
almost identical to the situation where no anti-conformist is present.

\item Suppose that the conformists have very small $l^c,r^c$ ($<\epsilon$),
   which means that they react to the anti-conformists. Then most of the
    classes become impossible, in particular $N^a,N^c$, and only  (13) (if
  $r^a$ is large enough), and (14) (if $l^a$ is large
  enough) remain. Otherwise, we get $2^N$ (see Figure~\ref{fig:4}).
\begin{figure}[htb]
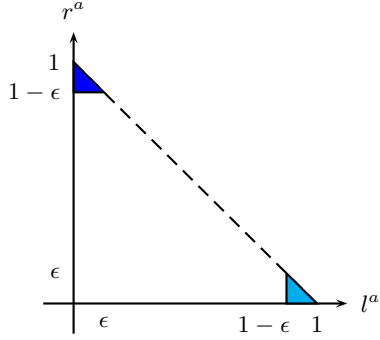

\begin{center}
\psset{unit=0.8cm}
\pspicture(-0.5,-1)(4.5,5)
\psline{->}(-0.5,0)(4.5,0)
\psline{->}(0,-0.5)(0,4.5)
\uput[0](4.5,0){$l^a$}
\uput[-90](4,0){\small $1$}
\uput[-120](3.5,0){\small $1-\epsilon$}
\uput[-90](0.5,0){\small $\epsilon$}
\uput[90](0,4.5){$r^a$}
\uput[180](0,4){\small $1$}
\uput[180](0,3.5){\small $1-\epsilon$}
\uput[180](0,0.5){\small $\epsilon$}
\psline[linestyle=dashed](0,4)(4,0)
\pspolygon[fillstyle=solid,fillcolor=blue](0,3.5)(0,4)(0.5,3.5)
\pspolygon[fillstyle=solid,fillcolor=cyan](3.5,0)(4,0)(3.5,0.5)
\endpspicture
\caption{Phase diagram for Situation 3, with $l^c,r^c<\epsilon$. Color code:
  blue=chaotic `no' (13), cyan=chaotic `yes' (14),
white=chaos (20)}
\label{fig:4}
\end{center}
\end{figure}
In this case, no consensus is possible, even in a weak sense, and only chaotic
situations arise. 
\end{itemize}

\section{Related literature}\label{sec:related-literature}

In this section we briefly mention some related literature different from our previous works recalled in Section \ref{sec:anoinf}. 

Opinion conformity has been studied widely in various fields and settings, and by using different approaches; for surveys, see e.g., \cite{jac08a,ace-ozd11}. A subset of this literature focuses on various extensions of the DeGroot model (\cite{deg74}), see e.g., 
\cite{dem03,jac08a,gol-jac10,bue-hel-pic14,bue-hel-klo14,gra-man-rus-tan17}, and for a survey, e.g., \cite{gol-sad16}. 
So far, the analysis of the anti-conformist behavior is much less common than the study devoted to the phenomenon of conformism. 

\cite{gra-rus07,gra-rus09a} address the problem of measuring negative influence in a social network but only in one-step (static) settings. \cite{bue-hel-klo14} study a dynamic model of opinion formation, where agents update their opinion by averaging over opinions of their neighbors, but might misrepresent their own opinion by conforming or counter-conforming with the neighbors. Although their model is related to \cite{deg74}, it is very different from our framework of anonymous influence with conformist and anti-conformist agents. Moreover, the authors focus on the relation between an agent's influence in the long run opinion and network centrality, and on wisdom of the society, while we determine all possible absorbing classes and conditions for their occurrence. 

\cite{kon-leb-web97} present a setting completely different from the present paper but related to our definition of anti-conformist agents. They consider a non-cooperative anonymous game in which one of the assumptions on individuals' preferences is partial rivalry, implying that the payoff of every player increases if the number of players who choose the same strategy declines. The authors examine the existence of strong Nash equilibrium in pure strategies for such a game with a finite set of players, and then with continuum of players. 

There are several other works that study network formation and anti-coordination games, i.e., games where agents prefer to choose an action different from that chosen by their partners. Our approach is different from anti-coordination games, in particular, because we have an essential dissymmetry between agents. \cite{bra07} investigates anti-coordination games played on fixed networks. In his model, agents are embedded in a fixed network and play with each of their neighbors a symmetric anti-coordination game, like the chicken game. The author examines how social interactions interplay with the incentives to anti-coordinate, and how the social network affects choices in equilibrium. He shows that the network structure has a much stronger impact on the equilibria than in coordination games.    
\cite{bra-lop-goy-veg04} study anti-coordination games played on endogenous networks, where players choose partners as well as actions in coordination games played with their partners. They characterize (strict) Nash architectures and study the effects of network structure on agents' behavior. The authors show that both network structure and induced behavior depend crucially on the value of cost of forming links. 
\cite{lop09} extends the model of \cite{bra-lop-goy-veg04} which is one-sided to a framework in which the cost of link formation is not necessarily distributed as in the one- or two-sided models, but is shared between the two players forming the link. She introduces an exogenous parameter specifying the partition of the cost and characterizes the Nash equilibria depending on the cost of link formation and the cost partition. 

\cite{koj-tak07} introduce the class of anti-coordination games and investigate the dynamic stability of the equilibrium in a one-population setting. They focus on the best response dynamic where agents in a large population take myopic best responses, and the perfect foresight dynamic where agents maximize total discounted payoffs from the present to the future. 

\cite{cao-gao-qu-yan-yan13} consider the fashion game of pure competition and pure cooperation. It is a network game with conformists (`what is popular is fashionable') and rebels (`being different is the essence') that are located on social networks (a spatial cellular automata network and small-world networks). The authors run simulations showing that in most cases players can reach a very high level of cooperation through the best response dynamic. They define different indices (cooperation degree, average satisfaction degree, equilibrium ratio and complete ratio) and apply them to measure players' cooperation levels. 

Our setting can be applied to some existing models, like herd behavior and information cascades (\cite{ban92,bik-hir-wel92}) which have been used to explain fads, investment patterns, etc.; see \cite{and-hol08} for a survey of experiments on cascade behavior. Although \cite{bik-hir-wel92} have already addressed the issue of fashion, the present model takes a different turn, since we assume no sequential choices and some agents are anti-conformists while others are conformists. In the model of herd behavior (\cite{ban92}) agents play sequentially and wrong cascades can occur. Though it can be rational to follow the crowd, some anti-conformists may want to play a mixed-strategy: either following the crowd or not. This is particularly true under bounded rationality. Agents may not be able to know what is rational, for example because
they lack information or do not have enough time or computational capacities. As a consequence, they may play according to rules of the thumb like counting how many people said `yes' rather than computing bayesian probabilities. \cite{cha-lar-xan16} show in a lab experiment that people tend to behave according to the DeGroot model rather than to Bayesian updating; see also \cite{cel-kar04}. This is also consistent with \cite{and-hol97} who show that counting is the most salient bias to explain departure from Bayesian updating. 

\section{Concluding remarks}\label{sec:conclusion}

Clearly, the present paper has taken a different road than the references mentioned in Section \ref{sec:related-literature}. We analyze a process of opinion formation in a society with different types of agents: pure conformists, pure anti-conformists, and mixed agents. We focus on anonymous influence, where a change of an agent's opinion depends on the number of agents with a certain opinion and not on their identities. We determine all possible absorbing classes and conditions for their occurrence for the society without mixed agents as well as for the mixed case. Moreover, the analysis of a very large society in different types of situations is provided. 

First of all, our study confirms and puts in precise terms what the intuition
says to us: the introduction of anti-conformists in a society, even in a very
small proportion, prevents from reaching a consensus and causes either
polarization or various instabilities: cycles, chaotic behavior, etc. Our study
has shown that, even under some simplifying assumptions (the parameters
$l^c,r^c,l^a,r^a$ are supposed to be the same for every agent in a category), the
convergence issue is very complex and many (up to 20) different situations can
occur. Despite this apparent complexity, we have managed to draw some general
and instructive conclusions which are valid in different typical situations. We
summarize below our main findings of Section~3, established in the pure case
(that is, agents are either purely conformist or purely anti-conformist) and
with a society of large size:
\begin{itemize}
\item In a society where all agents have the same influenceability
  characteristics, a cascade effect leading to a polarization is likely to
  occur. The type of polarization depends on the firing threshold, i.e., the
  proportion of `yes' which is necessary to start being influenced. If the
  firing threshold is low, then all conformist agents will finally say `yes',
  while if the firing threshold is high, the cascade effect leads all
  conformists to say `no'. This cascade phenomena happen even with a very small
  number of anti-conformists, and tend to be blurred by a chaotic behavior if
  the proportion of anti-conformists becomes large. It shows a very important
  fact: {\it the opinion of a society can be manipulated by introducing a small
    proportion of anti-conformists in it (opinion reversal).} Hence, anti-conformists do not only
  introduce chaotic behavior, they can steer the opinion in some direction.
\item When agents have a symmetric behavior w.r.t. `yes' or `no' in terms of
  influenceability, no cascade phenomenon can occur, and a chaotic behavior is
  very likely. A polarization can occur however, if the anti-conformists are not
  ``seen'' by the conformists (i.e., their number stays below the firing
  threshold), and all the more since the anti-conformists are reactive.
\item When the proportion of anti-conformists becomes very small, and if they
  are not ``seen'' by the conformists, then the situation is as if there were no
  anti-conformists at all (this shows a kind of continuity property of the
  model). If on the contrary they can be seen, some cascade effect
  is possible (precisely, only if either $l^a$ or $r^a$ is smaller
  than the proportion of anti-conformists).
\item Lastly we mention a special situation similar to a triple point in
  physics: the three types of behavior (polarization, fuzzy polarization and
  cycle) coexist. This can happen if and only if half of the population is
  anti-conformist and the firing threshold of conformists and anti-conformists
  is equal to 1/2.
\end{itemize}  
The introduction of mixed agents has clear effects on opinion formation. Mixed
agents do not change the number of possible absorbing classes, but their
presence blurs them, because the opinion of mixed agents always
oscillate between conformism and anti-conformism. This means that neither $N^a$
nor $N^c$ can appear as absorbing states or be a constituent of an absorbing
class, but they are replaced by their blurred version, where any subset of mixed
agents can be present. As a consequence, polarization {\it stricto sensu} cannot
appear anymore.

\section{Acknowledgments}
The authors thank the Agence Nationale de la Recherche for
financial support under contracts ANR-13-BSHS1-0010 (DynaMITE) and ANR-10-LABX-93-01 (Labex OSE).

\bibliographystyle{plainnat}
\bibliography{rusinowska,references}

\appendix
\section{Proof of Theorem~\ref{th:main}}
Our strategy is based on (F6): aperiodic absorbing classes are connected collections $\cS$
such that $\cS\srightarrow \cS$. Periodic absorbing classes are of the form
$\cS_1\srightarrow\cdots\srightarrow \cS_p$ with all $\cS_i$ pairwise
incomparable, and $\cS_1\cup\cdots\cup\cS_p$ is connected. Consequently, we
study all possible kinds of transition $\cS\srightarrow \cT$, and check
connectedness for each candidate. We distinguish between ``simple'' transitions
of the type $\cB\srightarrow \cB'$ with $\cB,\cB'\in\BB$, and ``multiple''
transitions $\cS\srightarrow \cT$, where $\cS,\cT$ are composed with several
elements of $\BB$, e.g., $[\emptyset,N^a]\cup[\emptyset,N^c]$. 

\subsection{Simple transitions}
We focus on transitions of the type $\cB\srightarrow\cB'$, with
$\cB,\cB'\in\BB$, and look for conditions on the parameters of the model to
obtain such transitions. 

Observe that if $\cB'$ is a nontrivial interval, it cannot be the union of other
elements of $\cB$. Therefore, $\cB\srightarrow \cB'$ if and only if for any $S\in\cB$, $S\srightarrow \cB''$ with $\cB''\in\BB$ and $\cB''\subseteq
\cB'$, and there is at least one $S\in \cB$ s.t. $S\srightarrow \cB'$. Let us
denote by $\cC[\cB]$ the conditions on $s=|S|$ to have a sure transition from
$S$ to $\cB$, as given in Table~\ref{tab:1}. All these conditions are intervals.

Observe that all $\cB\in\BB$ are either singletons $\{B\}$ or nontrivial
intervals $[\uB,\oB]$, and $\cB\subset \cB'$ if and only if $\cB=\{\uBp\}$ or
$\{\oBp\}$, with $\cB'=[\uBp,\oBp]$. Hence:
\begin{equation}\label{eq:4}
\cB\srightarrow
\cB'\Leftrightarrow \begin{cases}[\ub,\ob]\subseteq\cC[\cB']\cup\cC[\{\uBp\}]\cup\cC[\{\oBp\}]
  \\  [\ub,\ob]\cap\cC[\cB']\neq\emptyset,\end{cases}
\end{equation}
with $\ub,\ob$ the cardinalities of $\uB,\oB$. Let us apply (\ref{eq:4}) to all
possibilities. When $\{\cB'\}$ is a singleton, the above condition reduces to
$[\ub,\ob]\subseteq\cC[\cB']$, as given in Table~\ref{tab:1}. Otherwise, 
\begin{enumerate}
\item with $\cB'=[\emptyset,N^a]$, we obtain $[\ub,\ob]\subseteq[0,l^c]$ and
  $[\ub,\ob]\cap \left]l^a,n-r^a\right[\cap [0,l^c]\neq\emptyset$, which
    simplifies to
\begin{equation}\label{eq:st1}
[\ub,\ob]\subseteq[0,l^c] \text{ and }
  [\ub,\ob]\cap \left]l^a,n-r^a\right[\neq\emptyset;
\end{equation}
\item with $\cB'=[\emptyset, N^c]$, we obtain 
\begin{equation}\label{eq:st2}
[\ub,\ob]\subseteq[n-r^a,n] \text{ and }
  [\ub,\ob]\cap \left]l^c,n-r^c\right[\neq\emptyset;
\end{equation}
\item with $\cB'=[N^c,N]$, we obtain
\begin{equation}\label{eq:st3}
[\ub,\ob]\subseteq[n-r^c,n] \text{ and }
[\ub,\ob]\cap \left]l^a,n-r^a\right[\neq\emptyset;
\end{equation}
\item with $\cB'=[N^a,N]$, we obtain 
\begin{equation}\label{eq:st4}
[\ub,\ob]\subseteq[0,l^a] \text{ and }
  [\ub,\ob]\cap \left]l^c,n-r^c\right[\neq\emptyset.
\end{equation}
\end{enumerate}
This yields Table~\ref{tab:3}. Observe that the table is symmetric w.r.t. its center by the symmetry principle (Lemma~\ref{lem:sym}): just exchange $r$ with $l$. 
The transitions being sure, all cases on each line are exclusive.

From Table~\ref{tab:3}, we can deduce absorbing classes reduced to singletons or
intervals: they correspond to transitions $\cS\srightarrow\cS$ in the table,
provided they are connected. We obtain:
\begin{enumerate}
\item $N^a$, under the condition $n^c\geq (n-l^c)\vee(n-l^a)$;
\item $N^c$, under the condition $n^c\geq (n-r^c)\vee(n-r^a)$;
\item $[\emptyset,N^a]$, under the condition $ n-l^c  \leq n^c <n-l^a$;
\item $[N^c,N]$, under the condition $n-r^c\leq n^c<n-r^a$.
\end{enumerate}
We check connectedness for (iii) ((iv) follows by symmetry). We see from
Table~\ref{tab:1} that every $S\in[\emptyset, N^a]$ with $s\leq l^a$ has a sure
transition to $N^a$, while the other ones go to every set in the
interval. Therefore, the interval is connected if and only if $N^a$ has
a possible transition to every set in the interval, i.e., we need
$l^a<n^a<n-r^a$ and $n^a\leq l^c$, so the additional condition $n^a<n-r^a$ is
needed. In summary:
\begin{enumerate}
\item $N^a$ is an absorbing class if and only if $n^c\geq (n-l^c)\vee(n-l^a)$;
\item $N^c$ is an absorbing class if and only if $n^c\geq (n-r^c)\vee(n-r^a)$;
\item $[\emptyset,N^a]$ is an absorbing class if and only if $(n-l^c)\vee(r^a+1)  \leq n^c <n-l^a$;
\item $[N^c,N]$ is an absorbing class if and only if $(n-r^c)\vee(l^a+1)\leq n^c<n-r^a$.
\end{enumerate}

In order to get (absorbing) cycles and periodic classes, we study chains of sure
transitions of length 2: $\cS_1\srightarrow \cS_2\srightarrow \cS_3$, with
$\cS_1,\cS_2,\cS_3$ being pairwise disjoint, except possibly $\cS_1=\cS_3$. An
inspection of Table~\ref{tab:3} yields all such possible chains of length
2, summarized in Table~\ref{tab:4}. A second table can be obtained by symmetry.

From Table~\ref{tab:4}, we obtain the following candidates for absorbing cycles and periodic
classes, after eliminating double occurrences and using symmetry:
\begin{enumerate}
\item  $N^a\srightarrow \emptyset\srightarrow N^a$, under the condition
  $n-l^c\leq n^c\leq r^a$; 
\item  $N^c\srightarrow N\srightarrow N^c$, under the condition
  $n-r^c\leq n^c\leq l^a$; 
\item $N^c\srightarrow N^a\srightarrow N^c$, under the condition $n^c\leq
  l^c\wedge l^a\wedge r^c\wedge r^a$;
\item  $[\emptyset, N^c]\srightarrow N^a\srightarrow [\emptyset,N^c]$, under the condition
         $n^c \leq l^c\wedge l^a\wedge
           r^a,r^c<n^c<n-l^c$
\item $[N^a,N]\srightarrow N^c\srightarrow [N^a,N]$, under the condition $n^c
  \leq r^c\wedge r^a\wedge l^a,l^c<n^c<n-r^c$
\item $[N^a,N]\srightarrow [\emptyset,N^c]\srightarrow [N^a,N]$,
             under the condition $r^c \vee l^c <n^c \leq r^a \wedge l^a $.
\end{enumerate}
It remains to check connectedness of (iv) and (vi) ((v) is obtained by
symmetry). For (iv), we must check that $N^a$ has a possible transition to every
set in $[\emptyset,N^c]$. By Table~\ref{tab:1}, we must have $n^a\geq n-r^a$ and
$l^c<n^a<n-r^c$, which is true by the conditions in (iv). We address (vi). We
claim that under the conditions in (vi) $[N^a,N]\cup[\emptyset,N^c]$ is
connected if and only if $N^a\srightarrow [\emptyset,N^c]$ and
$N^c\srightarrow[N^a,N]$. Take any $S\in[\emptyset,N^c]$. Then $S$ goes either
to any set $T$ in $[N^a,N]$ or only to $N^a$ or only to $N$. In the first case,
similarly, $T$ goes either to any set $S'\in[\emptyset,N^c]$ (and we are done)
or only to $\emptyset$ or only to $N^c$. If $T\srightarrow \emptyset$, then we
have $T\srightarrow \emptyset\srightarrow N^a\srightarrow[\emptyset,N^c]$ and we
are done. Otherwise we have $T\srightarrow N^c\rightarrow
N^a\srightarrow[\emptyset,N^c]$. Suppose now that $S\srightarrow N^a$, then
$N^a$ goes to any $S'\in[\emptyset,N^c]$ and we are done. Otherwise,
$S\srightarrow N\srightarrow N^c\rightarrow N^a\srightarrow[\emptyset,N^c]$ and
we are done. This proves sufficiency. Now suppose the condition is not
fulfilled. This means that $N^a$ goes to either $\emptyset$ or $N^c$ (or similar
condition for $N^c$). In fact, due to the conditions in (vi) and
Table~\ref{tab:1}, we have that $N^a\srightarrow \emptyset$, but this yields the
cycle $N^a\srightarrow \emptyset\srightarrow N^a$. 

So in summary, candidates from (i) to (v) are all periodic classes under
the specified conditions, and for (vi), the additional condition that
$N^a\srightarrow [\emptyset,N^c]$ and $N^c\srightarrow[N^a,N]$ yields:
\begin{itemize}
\item[(vi')]$[N^a,N]\srightarrow [\emptyset,N^c]\srightarrow [N^a,N]$
             under the condition $r^c \vee l^c <n^c \leq r^a \wedge l^a \wedge
             (n-l^c-1)\wedge (n-r^c-1)$.
\end{itemize}

For cycles and periodic classes of length 3, by combining the possible chains of
length 2 of Table~\ref{tab:4} with possible transitions of Table~\ref{tab:3}, we
have only one candidate, all other being eliminated because the collections are not disjoint:
\[
N^c\srightarrow \emptyset \srightarrow N^a\srightarrow N^c.
\]
Hence we find, taking
into account the symmetry, two additional cycles:
\begin{enumerate}
\item $N^c\srightarrow \emptyset \srightarrow N^a\srightarrow N^c$, under the
  condition $n^c\leq r^c\wedge r^a\wedge l^c,n^c\geq n-r^a$;
\item $N^a\srightarrow N \srightarrow N^c\srightarrow N^a$, under the
  condition $n^c\leq l^c\wedge l^a\wedge r^c,n^c\geq n-l^a$.
\end{enumerate}

We now show that periodic classes of period greater than three cannot
exist, which finishes the study of simple transitions.
\begin{Lemma}
There exists no periodic class of period $k\geq 4$.
\end{Lemma}
\begin{proof}
 Let $\cS$ be a periodic class. First, observe that if $\emptyset,N$ are not
elements of $\cS$, it is not possible to choose four distinct elements of
$\mathbb{B}\setminus\{\{ \emptyset\},\{N\}\}$ such that these elements are
pairwise disjoint. Hence, we suppose that there are transitions
$\cB\srightarrow\emptyset$ and/or $\cB\srightarrow N$ in $\cS$. From
Table~\ref{tab:3}, we see that $\cB$ is necessarily $\{N^a\}$ or $\{N^c\}$.

We claim that the cycle $\emptyset \xrightarrow{1} N^a \xrightarrow{1} N
\xrightarrow{1}N^c \xrightarrow{1} \emptyset $ is impossible. Indeed, by
Table~\ref{tab:4}, we have $\emptyset \xrightarrow{1} N^a \xrightarrow{1} N$ iff
$n-l^a\leq n^c\leq r^c$ and $N \xrightarrow{1}N^c \xrightarrow{1} \emptyset $
(its symmetric) iff $n-r^a\leq n^c\leq l^c$. This yields, respectively,
\begin{align*}
2n^c &\geq 2n-l^a-r^a>n\\
2n^c &\leq r^c+l^c<n,
\end{align*}
a contradiction.

Assume that we have a transition to $\emptyset$ (the case for $N$ is obtained by
symmetry). We have either $N^a\srightarrow\emptyset$ (which is discarded because
it leads to the cycle $N^a\srightarrow \emptyset\srightarrow N^a$) or
$N^c\srightarrow \emptyset$. Then, the only possible absorbing class of the
  form $ N^c \xrightarrow{1} \emptyset \xrightarrow{1} N^a \xrightarrow{1} \cB_1
  \xrightarrow{1} \cdots \xrightarrow{1} \cB_p \xrightarrow{1} N^c $ is the
  cycle $\emptyset \xrightarrow{1} N^a \xrightarrow{1} N^c \xrightarrow{1}
  \emptyset $, for, either $\cB_1=N$, and we obtain the impossible cycle in the
  claim above, or $\cB_1$ contains $N^a$ or $N^c$, which is impossible since
  elements in $\cS$ should be pairwise disjoint.
\qed\ 
\end{proof}

\subsection{Multiple transitions}
We examine the case of transitions of the form $\cS\srightarrow \cB_1\cup\cdots
\cup\cB_p$, with $p\geq 2$, $\cS\in 2^N$ and formed only from sets in $\BB$,
$\cB_1,\ldots,\cB_p\in\BB$, and all $\cB_1,\ldots,\cB_p$ are pairwise
incomparable by inclusion\footnote{The ``$\cup$'' is understood at the level of
  collections of sets, i.e., $\cB_1\cup\cB_2=\{S\in 2^N\mid S\in\cB_1\text{ or }
  S\in\cB_2\}$.}. The analysis is done in the same way as for simple
transitions: the above transition exists if and only if for every $S\in\cS$,
$S\srightarrow \cB'$ with $\cB'\in\BB$ and $\cB'\subseteq \cB_1\cup\cdots\cup
\cB_p$ and there exist distinct $S_1,\ldots, S_p\in\cS$ such that
$S_j\srightarrow\cB_j$ for $j=1,\ldots, p$, which readily shows that $\cS$
cannot be a singleton. More explicitly, using previous notation and denoting by
$\supp(\cS)=\{|S|\ : \ S\in\cS\}$ the support of $\cS$, we get:
\begin{equation}\label{eq:5}
\cS\srightarrow
\cB_1\cup\cdots\cup\cB_p\Leftrightarrow \begin{cases}\displaystyle\supp(\cS)\subseteq\bigcup_{j=1}^p\cC[\cB_j]\cup\bigcup_{j=1}^p\cC[\{\uB_j\}]\cup\bigcup_{j=1}^p\cC[\{\oB_j\}]
  \\ \displaystyle \supp(\cS)\cap\cC[\cB_j]\neq\emptyset, \ \ j=1,\ldots,p.\end{cases}
\end{equation}
Let us investigate what the possible candidates for $\cB_1\cup\cdots\cup
\cB_p$ are. We begin by restricting to nontrivial intervals and $p=2$. From
Table~\ref{tab:1}, we find:
\begin{enumerate}
\item $\cS\srightarrow[\emptyset,N^a]\cup[\emptyset,N^c]$ if and only if
\begin{equation}\label{eq:mt1}
\supp(\cS)\subseteq [0,l^c]\cup [n-r^a,n]\text{ and }\begin{cases} \supp(\cS)\cap
\left]l^a,n-r^a\right[\cap[0,l^c]\neq\emptyset \\ 
\supp(\cS)\cap\left] l^c,n-r^c\right[\cap [n-r^a,n]\neq\emptyset\end{cases}; 
\end{equation}
\item $\cS\srightarrow[\emptyset,N^a]\cup[N^c,N]$ if and only if
\begin{equation}\label{eq:mt2}
\supp(\cS)\subseteq [0,l^c]\cup [n-r^c,n]\text{ and } \begin{cases}\supp(\cS)\cap
\left]l^a,n-r^a\right[\cap [0,l^c]\neq\emptyset\\
\supp(\cS)\cap\left]l^a,n-r^a\right[\cap [n-r^c,n]\neq\emptyset\end{cases}; 
\end{equation}
\item $\cS\srightarrow[N^a,N]\cup[\emptyset,N^c]$ if and only if
\begin{equation}\label{eq:mt3}
\supp(\cS)\subseteq [0,l^a]\cup [n-r^a,n]\text{ and } \begin{cases} \supp(\cS)\cap
\left]l^c,n-r^c\right[\cap[0,l^a]\neq\emptyset\\
\supp(\cS)\cap\left]l^c,n-r^c\right[\cap[n-r^a,n]\neq\emptyset\end{cases};
\end{equation}
\item  $\cS\srightarrow[N^a,N]\cup[N^c,N]$ if and only if
\begin{equation}\label{eq:mt4}
\supp(\cS)\subseteq [0,l^a]\cup [n-r^c,n]\text{ and } \begin{cases} \supp(\cS)\cap
\left]l^c,n-r^c\right[\cap[0,l^a]\neq\emptyset\\
\supp(\cS)\cap\left]l^a,n-r^a\right[\cap[n-r^c,n]\neq\emptyset\end{cases},
\end{equation}
\end{enumerate}
the other combinations $[\emptyset,N^a]\cup[N^a,N]$ and
$[\emptyset,N^c]\cup[N^c,N]$ being impossible as it can be checked. This readily
shows that $p>2$ with nontrivial intervals is impossible since a forbidden
combination would appear in the list.

We consider now that singletons may appear. We begin by noticing that there is no
absorbing class of the form $\{S_1,\ldots, S_p\}$ with
$S_j\in\{\emptyset,N,N^a,N^c\}$ for all $j$ and $p\geq 2$. Indeed, Table~\ref{tab:3} shows
that transitions from a set $S$ can only lead to a single $T$, with no
possibility of multiple transition. Hence, such collections would never be
connected.

Let us examine the case $\cS\srightarrow\cB_1\cup\{S\}$, where $\cB_1$ is a
nontrivial interval. With $[\emptyset,N^a]\cup\{N\}$ we obtain:
\[
\supp(\cS)\subseteq [0,l^c]\cup([0,l^a]\cap[n-r^c,n]) \text{ and }\begin{cases}
\supp(\cS)\cap[0,l^c]\cap \left]l^a,n-r^a\right[\neq\emptyset\\
\supp(\cS)\cap[0,l^a]\cap [n-r^c,n]\neq\emptyset\end{cases},
\]
which is impossible. With $[\emptyset,N^a]\cup\{N^c\}$ we obtain
\begin{equation}\label{eq:mt5}
\supp(\cS)\subseteq [0,l^c]\cup([n-r^a,n]\cap[n-r^c,n]) \text{ and }\begin{cases}
\supp(\cS)\cap[0,l^c]\cap \left]l^a,n-r^a\right[\neq\emptyset\\
\supp(\cS)\cap[n-r^c,n]\cap [n-r^a,n]\neq\emptyset\end{cases},
\end{equation}
which is possible.  Similarly, we find that $[\emptyset,N^c]\cup\{N\}$,
$[N^a,N]\cup\{\emptyset\}$ and $[N^c,N]\cup\{\emptyset\}$ are impossible, while
the following are possible:
\begin{enumerate}
\item $\cS\srightarrow[\emptyset,N^c]\cup\{N^a\}$ iff 
\begin{equation}\label{eq:mt6}
\supp(\cS)\subseteq [n-r^a,n]\cup([0,l^a]\cap[0,l^c]) \text{ and }\begin{cases}
\supp(\cS)\cap[n-r^a,n]\cap\left]l^c,n-r^c\right[\neq\emptyset\\
\supp(\cS)\cap[0,l^a]\cap [0,l^c]\neq\emptyset\end{cases},
\end{equation}
\item $\cS\srightarrow[N^a,N]\cup\{N^c\}$ iff 
\begin{equation}\label{eq:mt7}
\supp(\cS)\subseteq [0,l^a]\cup([n-r^a,n]\cap[n-r^c,n]) \text{ and }\begin{cases}
\supp(\cS)\cap[0,l^a]\cap \left]l^c,n-r^c\right[\neq\emptyset\\
\supp(\cS)\cap[n-r^a,n]\cap [n-r^c,n]\neq\emptyset\end{cases},
\end{equation}
\item $\cS\srightarrow[N^c,N]\cup\{N^a\}$ iff 
\begin{equation}\label{eq:mt8}
\supp(\cS)\subseteq [n-r^c,n]\cup([0,l^a]\cap[0,l^c]) \text{ and }\begin{cases}
\supp(\cS)\cap[n-r^c,n]\cap\left]l^a,n-r^a\right[\neq\emptyset\\
\supp(\cS)\cap[0,l^a]\cap [0,l^c]\neq\emptyset\end{cases}.
\end{equation}
\end{enumerate}
This shows that transitions of the form
$\cS\srightarrow\cB\cup\{S_1\}\cup\{S_2\}$ are not possible since a forbidden
configuration would appear.

We are now in position to study aperiodic absorbing classes. 
\begin{enumerate}
\item With $\cS=[\emptyset,N^a]\cup[\emptyset,N^c]$, we find from (\ref{eq:mt1}) that
\[
[0,n^a\vee n^c]\subseteq [0,l^c]\cup [n-r^a,n] \text{ and } \begin{cases}
[0,n^a\vee n^c]\cap\left]l^a,n-r^a\right[\cap[0,l^c]\neq\emptyset\\
[0,n^a\vee n^c]\cap\left]l^c,n-r^c\right[\cap[n-r^a,n]\neq\emptyset\end{cases}
\]
which is equivalent to 
\begin{equation}\label{eq:p1}
n^a\vee n^c>l^c\geq n-r^a.
\end{equation}

We check connectedness. We begin by a simple observation. We have
$\emptyset\srightarrow N^a$, therefore we must forbid the transitions
$N^a\srightarrow\emptyset$ and $N^a\srightarrow N^a$. Using Table~\ref{tab:1}
and (\ref{eq:p1}), we find that
$n^a\in\left]l^a,n-r^a\right[\cup\left]l^c,n\right[$.  Suppose that $n^a
\in\left]l^a,n-r^a\right[$. From Table~\ref{tab:1}, we obtain that
$N^a\srightarrow[\emptyset,N^a]\srightarrow[\emptyset,N^a]$, hence no
connection to $[\emptyset,N^c]$ is obtained.  Therefore we are forced to
consider $n^a\in\left]l^c,n\right[$, which with (\ref{eq:p1}) leads to
\begin{equation}\label{eq:p2}
n^a>l^c\geq n-r^a.
\end{equation}
From Table~\ref{tab:1} again, this implies $N^a\srightarrow [\emptyset,N^c]$
when $n^a\in\left]l^c,n-r^c\right[$, or $N^a\srightarrow N^c$ when
    $n^a\in\left[n-r^c,n\right[$. We distinguish the two cases.

1. Suppose $n^a\in\left]l^c,n-r^c\right[$, so we have $\emptyset\srightarrow
        N^a\srightarrow[\emptyset,N^c]$. In order to connect $\emptyset,N^a$ to
        any set in $\left]\emptyset,N^a\right[$, there must exist
        $S\in[\emptyset,N^c]$ such that $S\srightarrow [\emptyset,N^a]$, i.e.,
        $s\in \left]l^a,n-r^a\right[\cap [0,l^c]=\left]l^a,n-r^a\right[$ by
        (\ref{eq:p2}). This is possible iff $n^c>l^a$.
Let us check whether $N^c$ is connected to any set in the
class. From Table~\ref{tab:1} and the condition $n^c>l^a$, we see that there is a possible transition to
$\emptyset$, which suffices to prove that $N^c$ is connected to any set in the
class, except if $n^c\in[n-r^c,n]$ in which case $N^c\srightarrow
N^c$. Therefore, we must ensure the following condition:
\begin{equation}\label{eq:p3}
n^c\in\left]l^a,n-r^c\right[.
\end{equation}
 We check similarly whether any other set in the class is connected  with
the rest.  Take $S\in\left]\emptyset,N^a\right[$. If $s\leq l^c$, there will be
either a possible transition to $\emptyset$ or to $N^a$, so that $S$ is
connected to any set in the class. If $s>l^c$, $S$ behaves like $N^a$ and we are
done. Take now $S\in\left]\emptyset,N^c\right[$. If $s\leq l^a$, then
$S\srightarrow N^a$ and we are done. If $s\in\left]l^a,l^c\right]$, $S$ has a
    possible transition to $\emptyset$ and we are done. Finally, if
    $s\in\left]l^c,n-r^c\right[$, $S$ behaves like $N^c$. In conclusion,
    (\ref{eq:p3}) summarizes the condition for connectedness in Case 1.

2. Suppose $n^a\in\left[n-r^c,n\right[$, so we have $\emptyset\srightarrow
        N^a\srightarrow N^c$. We must ensure that $N^c$ is connected to any set
        in the class. In order to avoid $N^c\srightarrow N^c$ and the
        transitions $N^c\srightarrow N^a$ and $N^c\srightarrow \emptyset$ which
        would lead to cycles, we are left with the cases
        $n^c\in\left]l^a,n-r^a\right[$ (yielding $N^c\srightarrow
    [\emptyset,N^a]$) and $n^c\in \left]l^c,n-r^c\right[$ (yielding
    $N^c\srightarrow[\emptyset,N^c]$). We examine both cases.

2.1. Suppose  $n^c\in\left]l^a,n-r^a\right[$, then we have $N^c\srightarrow
    [\emptyset,N^a]$. It remains to ensure that there exists
    $S\in\left]\emptyset,N^a\right[$ which is connected with
    $[\emptyset,N^c]$. We must have $s\in\left]l^c,n-r^c\right[$, always possible
      under Case 2. So we have established that $\emptyset,N^a,N^c$ are
      connected with the rest of the class. It remains to check if this is true
      for the other sets in the class. Take $S\in\left]\emptyset,N^a\right[$. If
        $s\leq l^c$, a transition to $\emptyset$ of $N^a$ is possible, and so we
        are done. If $s\in\left]l^c,n\right[$, then $S\rightarrow N^c$, and we
        are done. Take now $S\in\left]\emptyset,N^c\right[$. Then
        $s\in\left]0,n-r^a\right[$, so that $S\rightarrow N^a$ and we are
        done. As a conclusion, connectedness holds when
        $n^c\in\left]l^a,n-r^a\right[$. 

2.2. Suppose $n^c\in \left]l^c,n-r^c\right[$, then
        $N^c\srightarrow[\emptyset,N^c]$. It remains to connect some set $S$ in
    $\left]\emptyset,N^c\right[$ to $[\emptyset,N^a]$, which is possible iff
    $s\in\left]l^a,n-r^a\right[$. This is possible under Case 2, so $N^c$ is
    connected to any set in the class. We check for the remaining sets. Take
    $S\in\left]\emptyset,N^a\right[$. If $s\leq l^c$, a connection is possible
    to $N^a$ or $\emptyset$ so we are done. Otherwise, a connection to $N^c$ is
    possible and we are done. For $S\in\left]\emptyset,N^c\right[$, it works
    exactly the same. 

In conclusion of Case 2, connectedness is ensured iff
$n^c\in\left]l^a,n-r^a\right[\cup \left]l^c,n-r^c\right[$.

There does not seem to be a simple way to write the final condition. Here is one
possible:
connectedness holds iff $l^c\geq n-r^a$ and
\[
n^c\in\big(\left]r^c,n-l^c\right[\cap
    \left]l^a,n-r^c\right[\big)\cup\big(\big(\left]l^a,n-r^a\right[\cup\left]l^c,n-r^c\right[\big)\cap
    \left]0,r^c\right]\big). 
\]
\item Similarly, using (\ref{eq:mt2}), $\cS=[N^a,N]\cup[N^c,N]$ is an absorbing
  class if and only if $l^a \geq n-r^c $ and $ n^c \in \big(\left]l^c,n-r^c\right[\cap
  \left]r^a,n-l^c\right[\big)\cup\big(\big(\left]r^a,n-l^a\right[\cup\left]r^c,n-l^c\right[\big)\cap
  \left]0,l^c\right]\big)$.
\item With $\cS=[\emptyset,N^c]\cup[N^a,N]$ we find from (\ref{eq:mt3}) the
  condition $l^c\vee r^c<n^c\leq l^a\wedge r^a$. Let us check
  connectedness. Starting from $\emptyset$, we have $\emptyset\srightarrow
  N^a$, and by Table~\ref{tab:1} and the above condition we have
  $N^a\srightarrow [\emptyset,N^c]$ if $n^a>l^c$, and $N^a\srightarrow
  \emptyset$ otherwise. Clearly, the latter must be forbidden otherwise a cycle
  occurs. Therefore, we must have $n^a>l^c$. Moreover, we have $N^c\srightarrow
  [N^a,N]$ if $n^c<n-r^c$ and $N^c\srightarrow N$ otherwise. Since
  $N\srightarrow N^c$, the latter must be forbidden to avoid a cycle. Therefore,
  we must have $n^c<n-r^c$. Under these condition, from $\emptyset$ or $N^a$ or
  $N^c$, any set can be attained. Now, taking $S\in\left]\emptyset,N^c\right[$,
    we have $S\srightarrow N^a$ or $S\srightarrow[N^a,N]$ so that $S\rightarrow
    N^a$ and we are done. Lastly, taking $S\in\left]N^a,N\right[$, we have
    $S\srightarrow N^c$ or $[\emptyset,N^c]$ and we are done. As a conclusion,
    the condition is $l^c\vee r^c<n^c\leq l^a\wedge r^a$ and
    $n^c<(n-l^c)\wedge(n-r^c)$, but then we obtain the periodic absorbing class
    studied before. Indeed, we see from the proof that we have necessarily
    $[\emptyset,N^c]\srightarrow[N^a,N]\srightarrow[\emptyset,N^c]$. 
\item With $\cS=[\emptyset,N^a]\cup[N^c,N]$, using (\ref{eq:mt2}), we find that $l^a\vee r^a<
  n^a\leq l^c\wedge r^c$.  Suppose first $l^c+r^c<n-1$. Then  $\cS$ cannot be
  connected. Indeed, starting from $N^a$, we have from Table~\ref{tab:1} that
  for any set $S\in[\emptyset,N^a]$, we have either $S\srightarrow N^a$, or $S\srightarrow [\emptyset, N^a]$ or $S\srightarrow \emptyset$. Therefore, $[\emptyset,N^a]$
  is not connected with every set in $\cS$.

   Suppose now that $l^c+r^c=n-1$. The first condition in (\ref{eq:mt2})
    reduces to the void condition $\supp(\cS)\subseteq [0,n]$. By the second
    condition we deduce $l^a<l^c$ and $n-r^c<n-r^a$. We check connectedness by
    using Table~\ref{tab:1}. We must have $N^a\srightarrow[N^c,N]$, which
    happens iff $n-r^c\leq n^a< n-r^a$. Now, observe that for any
    $S\in[\emptyset,N^a]$, the transition is either in $[\emptyset, N^a]$ or
    $N^a$ (when $s\leq l^c$), or in $[N^c,N]$. To ensure that $[N^c,N]$ is
    connected to $[\emptyset, N^a]$, we must have $N^c\srightarrow[\emptyset,
      N^a]$, which happens iff $l^a<n^c\leq l^c$. Then any set $S\in[N^c,N]$ has
    a transition to either $[\emptyset, N^a]$ or $N^a$ (if $s\leq l^c$) or to
    $[N^c,N]$ or $N^c$. In summary, this class exists iff $l^c+r^c=n-1$,
    $n-r^c\leq n^a<n-r^a$ and $l^a<n^c\leq l^c$.  
\item We show that $[\emptyset,N^a]\cup\{N^c\}$ cannot be connected when
  $l^c+r^c\neq n-1$. Indeed, we must have $N^c\srightarrow [\emptyset,N^a]$
  or $N^c\srightarrow N^a$,
  which implies by Table~\ref{tab:1} the condition $n^c\leq l^c$. However, by
  (\ref{eq:mt5}) and the condition $l^c+r^c\neq n-1$, $\supp(\cS)$ must be in
  two disjoint intervals, implying that $[0,n^a]\subseteq [0,l^c]$ and
  $n^c\in[n-r^c,n]$, a contradiction.

We suppose now $l^c+r^c=n-1$ and $n-r^a\leq n-r^c$, so that in (\ref{eq:mt5})
the first condition reduces to the void condition $\supp(\cS)\subseteq [0,n]$.
 Observe that the second condition implies $l^c>l^a$. To ensure
  connectedness, we must have a transition from $N^a$ to $N^c$, which happens iff
  $n^a\geq n-r^c$. Also, we must ensure $N^c\srightarrow[\emptyset,N^a]$, which
  happens iff $l^a<n^c<n-r^a$. Finally, we must ensure that any
  $S\in[\emptyset,N^a]$ is such that either $S\srightarrow N^c$ or
      $S\srightarrow [\emptyset, N^a]$ or $S\srightarrow N^a$. The two latter
  transitions arise when $s\leq l^c$, while the former transition arises when
  $s\geq n-r^c$. Since $n-r^c=l^c+1$, no other transition can
  happen. Connectedness is then proved. Finally, it can be checked that the
  second condition in (\ref{eq:mt5}) is satisfied. In summary, this class exists
  iff $l^c+r^c=n-1$, $n-r^a\leq n-r^c$, $l^c>l^a$, $n^a\geq n-r^c$ and $l^a<n^c<n-r^a$.

\item With $[\emptyset,N^c]\cup\{N^a\}$, we find from (\ref{eq:mt6}) and the
  assumption $l^a+r^a\neq n-1$ that $\supp(\cS)$ must be in two disjoint
  intervals, which forces $n-r^a\leq n^a<n-r^c$ and $n^c\leq l^a\wedge l^c$. We
  know already that $[\emptyset,N^c]\srightarrow N^a \srightarrow [\emptyset,N^c]$
  is a periodic class. Let us show that this is the only possibility. Indeed,
  otherwise there should exist $S\in[\emptyset,N^c]$ such that
  $S\srightarrow[\emptyset, N^c]$. This would imply that $l^c<s<n-r^c$, which is
  impossible by the condition $n^c\leq l^c$.

Let us consider now that $l^a+r^a= n-1$ and $l^c\geq l^a$, so that in
(\ref{eq:mt6}) the first condition simply reduces to the void condition
$\supp(\cS)\subseteq [0,n]$, while the second becomes: either
$n^a\in\left]l^c,n-r^c\right[$ or $n^c>l^c$. Let us check connectedness. We must
    have $N^a\srightarrow [\emptyset, N^c]$ or $N^a\srightarrow N^c$. The first
    case happens iff $n^a\in\left]l^c,n-r^c\right[$. Then observe that without
    further condition on $n^c$, any set in $[\emptyset, N^c]$ is connected to
    either $N^a$, $\emptyset$, $[\emptyset,N^c]$ or $N^c$. It suffices then to
    forbid the transition $N^c\srightarrow N^c$, i.e., $n^c<n-r^c$. The second
    case happens iff $n^a\geq n-r^c$, which forces $n^c>l^c$. To ensure that
    $N^c$ is connected to $[\emptyset,N^c]$, we must have $l^c<n^c<n-r^c$. 
    Then any set in $[\emptyset, N^c]$ has a transition to either
    $N^a,\emptyset$ or $[\emptyset,N^c]$. In summary, this class exists iff
    $l^a+r^a= n-1$, $l^c\geq l^a$, and either $n^a\in\left]l^c,n-r^c\right[$ and
    $n^c<n-r^c$, or $n^a\geq n-r^c$ and $l^c<n^c<n-r^c$.
\item The case of $[N^a,N]\cup \{N^c\}$ is similar to its symmetric
  $[\emptyset,N^c]\cup\{N^a\}$.  The class exists iff $l^a+r^a= n-1$, $n-r^c\leq
  n-r^a$, and either $n^c\in\left]l^c,n-r^c\right[$ and $n^a>l^c$, or $n^c\leq
      l^c$ and $l^c<n^a<n-r^c$.
\item The case of $[N^c,N]\cup \{N^a\}$ is similar to its symmetric
  $[\emptyset,N^a]\cup \{N^c\}$. The class exists iff $l^c+r^c=n-1$,  $n-l^a\leq n-l^c$, $r^c>r^a$, $n^a\geq n-l^c$ and $r^a<n^c<n-l^a$.
\end{enumerate}

It remains to study the existence of periodic classes. Since the collections
must be pairwise disjoint, the only possibility is the periodic class
$[\emptyset,N^a]\cup[\emptyset,N^c]\srightarrow N\srightarrow
[\emptyset,N^a]\cup[\emptyset,N^c]$. But we know that the second transition is
impossible since a singleton cannot lead to a multiple transition. Hence, there
are no such periodic absorbing classes.

\begin{landscape}
\begin{table}[p]
\begin{center}
\small
\begin{tabular}{|c||c|c|c|c||c|c|c|c|}\hline
 $\cS \backslash\cT$ & $\emptyset$ & $N^a$ & $[\emptyset,N^c]$ & $[\emptyset,N^a]$ & $[N^c,N]$ &
  $[N^a,N]$ & $N^c$ & $N$ \\ \hline
$\emptyset$ & $\times$ & always & $\times$ & $\times$ &  $\times$ &  $\times$ &
  $\times$ & $\times$ \\ \hline
$N^a$ & $n-l^c\leq n^c\leq r^a$ & $\begin{array}{c}n^c\geq n-l^c\\n^c\geq n-l^a\end{array}$ & $\begin{array}{c} n^c
     \leq r^a\\r^c<n^c<n-l^c\end{array}$ & $\begin{array}{c} n^c\geq
       n-l^c\\ r^a<n^c<n-l^a\end{array}$ & $\begin{array}{c}
         n^c\leq r^c\\r^a<n^c<n-l^a\end{array}$ & $\begin{array}{c}
         n^c\geq n-l^a\\r^c<n^c<n-l^c\end{array}$ & $n^c\leq r^c\wedge r^a$ &
         $n-l^a\leq n^c\leq r^c$\\ \hline
$[\emptyset,N^c]$ & $\times$ & $n^c\leq l^c\wedge l^a$ & $\times$ &
         $l^a<n^c\leq l^c$ & $\times$ &
           $l^c<n^c\leq l^a$ & $\times$
             & $\times$ \\ \hline
$[\emptyset,N^a]$ & $\times$ & $\begin{array}{c}n^c\geq n-l^c\\n^c\geq
               n-l^a\end{array}$ & $\times$ & 
         $\begin{array}{c} n-l^c \leq n^c\\ n^c<n-l^a \end{array}$ & $\times$ &
           $\begin{array}{c} n-l^a\leq n^c\\ n^c<n-l^c \end{array}$ & $\times$
             & $\times$ \\ \hline\hline
$[N^c,N]$ & $\times$ & $\times$ & $\begin{array}{c} n-r^a \leq n^c\\ n^c<n-r^c \end{array}$ & $\times$ &
                     $\begin{array}{c} n-r^c \leq n^c\\ n^c <n-r^a \end{array}$
                       & $\times$ & $\begin{array}{c}n^c\geq n-r^c\\n^c\geq n-r^a\end{array}$ & $\times$ \\ \hline
$[N^a,N]$ & $\times$ & $\times$ & $r^c<n^c\leq r^a $ & $\times$ &
                     $r^a< n^c\leq r^c$
                       & $\times$ & $n^c\leq r^c\wedge r^a$ & $\times$ \\ \hline
$N^c$ & $n-r^a\leq n^c\leq l^c$ & $n^c\leq l^c\wedge l^a$ & $\begin{array}{c} n^c
     \geq n-r^a\\l^c<n^c<n-r^c\end{array}$ & $\begin{array}{c} n^c\leq
       l^c\\ l^a<n^c<n-r^a\end{array}$ & $\begin{array}{c}
         n^c\geq n-r^c\\l^a<n^c<n-r^a\end{array}$ & $\begin{array}{c}
         n^c\leq l^a\\l^c<n^c<n-r^c\end{array}$ & $\begin{array}{c}n^c\geq
           n-r^c\\n^c\geq n-r^a\end{array}$ &
         $n-r^c\leq n^c\leq l^a$\\ \hline
$N$ & $\times$ & $\times$ & $\times$ & $\times$ &  $\times$ &  $\times$ &
  always & $\times$ \\ \hline
\end{tabular}
\vspace{2 mm}
\caption{Conditions for sure transitions $\cS$ to $\cT$}
\label{tab:3}
\end{center}
\end{table}
\end{landscape}

\begin{table}[htb]
\begin{center}
\begin{tabular}{|c|c|}\hline
$N^a\srightarrow \emptyset\srightarrow N^a$ & $n-l^c\leq n^c\leq r^a$\\ \hline
$N^c\srightarrow \emptyset\srightarrow N^a$ & $n-r^a\leq n^c\leq l^c$\\ \hline\hline
$\emptyset\srightarrow N^a\srightarrow [N^c,N]$ & $\begin{array}{c}
         n^c\leq r^c\\r^a<n^c<n-l^a\end{array}$\\ \hline
$\emptyset\srightarrow N^a\srightarrow N^c$ & $n^c\leq r^c\wedge r^a$\\ \hline
$\emptyset\srightarrow N^a\srightarrow N$ & $n-l^a\leq n^c\leq r^c$\\  \hline
$N^c\srightarrow N^a\srightarrow \emptyset$ & $n-l^c\leq
         n^c\leq l^c\wedge l^a\wedge r^a$\\ \hline
 $[\emptyset, N^c]\srightarrow N^a\srightarrow [\emptyset,N^c]$ &
         $\begin{array}{c}n^c \leq l^c\wedge l^a\wedge
           r^a\\r^c<n^c<n-l^c\end{array}$\\ \hline 
$N^c\srightarrow N^a\srightarrow N^c$ & $n^c\leq l^c\wedge
         l^a\wedge r^c\wedge r^a$\\ \hline
$N^c$ or $[\emptyset, N^c]\srightarrow N^a\srightarrow N$ & $n-l^a\leq n^c\leq
         l^a\wedge l^c\wedge r^c$\\ \hline \hline
$N^a\srightarrow [\emptyset,N^c]\srightarrow N^a$ &
         $\begin{array}{c} n^c\leq l^a\wedge l^c\wedge
           r^a\\r^c<n^c<n-l^c\end{array}$\\ \hline
$[N^a,N]\srightarrow [\emptyset,N^c]\srightarrow [N^a,N]$
             & $l^c\vee r^c< n^c\leq r^a\wedge l^a$\\ \hline
\end{tabular}
\vspace{2 mm}
\caption{Conditions for chains of length 2 potentially yielding periodic classes}
\label{tab:4}
\end{center}
\end{table}

\end{document}